\documentclass[12pt]{article}

\usepackage[T1]{fontenc} 
\usepackage{amssymb,amsmath,amsthm,latexsym}
\usepackage{graphicx,float}
\usepackage[linesnumbered,ruled,vlined]{algorithm2e}
\usepackage{url}
\usepackage{tikz}

\newtheorem{definition}{Definition}
\newtheorem{lemme}{Lemma}
\newtheorem{example}{Example}

\author{  Stefi Nouleho \and Dominique Barth \and Franck Quessette \and Marc-Antoine Weisser \and Dimitri Watel \and Olivier David \\
 \normalsize DAVID, University of Versailles-Saint-Quentin \\
 \normalsize LRI, CentraleSupelec, Paris-Saclay University \\
 \normalsize ENSIIE, 1 square de la Resistance, 91025, Evry, France \\
 \normalsize SAMOVAR, Telecom SudParis, Evry \\
 \normalsize ILV, Department of Chemistry, University of Versailles, France 
}

\title {A new graph modelisation for molecule similarity}
\date{}

\begin{document}
\maketitle
\begin{abstract}

In order to define the process of restrosynthesis of a new organic molecule, it is often necessary to be able to draw inspiration from that of a  molecule similar to the target one of which we know such a process. To compute such a similarity, an oftently used approach is to solve a Maximum Common Edge Subgraph (MCES) problem \cite{Raymond2002} on molecular graphs, but such an approach is limited by computation time and pertinence of similarity measurement.  In this paper, we define and analyse here a new graph representation of molecules to algorithmically compare them. The purpose is to model the structure of molecule by a graph smaller than the molecular graph and representing the interconnexion of its elementary cycles. We provide an algorithm to efficiently obtain such a graph of cycles from a molecular graph. Then by solving MCES problems on those graphs, we evaluate the pertinence of using graphs of cycles for molecular similarity on a select set of molecules.

\end{abstract}

\section{Introduction}

In organic chemistry, when a new molecule is designed, it is necessary to determine chemical reactions that can be used to synthesize this target molecule from available compounds.  For this purpose, the goal we focus on here is to provide algorithms to help determine a sequence of reactions for the synthesis of new molecules,\textit{ i.e.} a tree of reactions whose root is the target molecule and  leaves are the available possibly initial molecules (commercialised or easily synthesizable).  In order to find such chemical reactions, chemists approach is to search in a reaction database (such as REAXYS \cite{reaxys} or CHEBI \cite{chebi}) for a molecule that is structurally close to the target molecule. And then from a chemical reaction of this similar molecule they draw inspiration reaction tree to propose such a chemical reaction for the target molecule. To help such a processus, it is therefore a question of being able to algorithmically select molecules in a reaction database that are structurally similar to a target molecule.

Considering a modeling of molecules by graphs or hypergraphs, several definitions and similarity approaches between molecules have already been studied \cite{Raymond2002}, mainly due to the principle stating that structurally similar molecules are expected to display similar properties \cite{Zager2008,Johnson1990},  or to help virtual screening for drug design  \cite{Eckert2007}. One of these approaches consists in  measuring distances of weighted editions between two molecular graphs, an edition being an operation of adding or deleting a vertex or a link in such a graph, or the label change of a vertex. These approaches are notoriously used in the field of bioinformatics \cite{Neuhaus2007,Sayle2015}. Another approach considers the kernel pattern of molecular graphs or hypergraphs, \textit{ie} the presence or not of sub-graphs (also called "fingerprints" \cite{Cereto2015,Bender2009}) in a set of determined patterns, close to cycles or trees, related to the functional properties of molecules; this approach seems well suited to the classification of molecules according to the properties concerned \cite{Grave2010,Gauzere2015}, but the choice of a significant set of substructures to compare molecules is ofently a difficult problem. Finally, a last approach considers the resolution of the problem of finding a Maximum Common Edge Subgraph\cite{Raymond2002}  (MCES) between two graphs. This problem is NP-complete and is initially seen as a generalization of graph isomorphism, with different metrics evaluating the size of this subgraph compared to those of the two graphs compared\cite{Faisal2007,Eckert2007,Zager2008,Akutsu2013}. It is a variant of this last approach that we will considered in this article.

In the context of this paper, we focus on the similarity of the structural configurations of two molecules. Such a  configuration is seen as the interconnexion of the cycles in the maximum 2-connected induced subgraph of the molecular graph. We assume that similar molecules may certainly have similar cyclic parts. A representation of the structure of a molecule based on the cycles it contains has already been proposed and used to classify and characterize sets of molecules \cite{Gauzere2013,Horvarth2004}. Here, we propose a definition of a cycle graph of a molecule modeling not only a relevant subset of the molecule cycles but also their interconnection, whether they share vertices or not; such a representation can also be seen as the extension of a reduction of the Markush structure of a molecule into a ring/non-ring reduction scheme leading to express the core structure of a molecule \cite{Lynch1996}, for example to make classification \cite{Gillet1991}. Our objective is to confirm that this definition of cycle graph corresponds sufficiently to the intuitive approach followed by a chemist and that the comparison of the graphs of cycles, based on a specific MCES, corresponds well to the notion of similarity of molecules wished.

The rest of the paper is organised as follow. In the next section, we give some preliminar definitions about graph theory and molecular graphs. Then in Section 3, we define the graph of cycle of a molecule and we propose an algorithm to efficiently obtain it for any given molecule. Finally, in Section 4, we evaluate the perfomances of using such graph of cycles (in terms of time computation and pertinence) to measure the similarity of pairs of molecules.
\section{Molecular graph}

In this paper, we consider definitions and notations on graph theory form Berge \cite{Berge1963}. An usual representation of a molecule is a molecular graph \cite{Gasteiger2003}. A molecular graph is an undirected labeled graph $G=(V,E)$ encoding the structural and functional information of the molecule. The set of vertices $V$ of $G$ encodes atoms and the set of edges $E$ encodes the adjacency relationship between atoms in the molecule. Each vertex is labeled by the corresponding chemical element (for example $C =$ Carbon, $H =$ Hydrogen) and each edge is labeled by its type of covalent bond (single $-$, double $=$, triple, aromatic).

Since hydrogen atoms can be connected at least to one atom, they can be omitted in the representation of a molecule (see Figure \ref{fig:graph}). A molecular graph does not encode neither the relative spatial arrangement of atoms nor the distance between atoms.

\begin{figure}[H]
\begin{center}
\includegraphics[scale = 0.41]{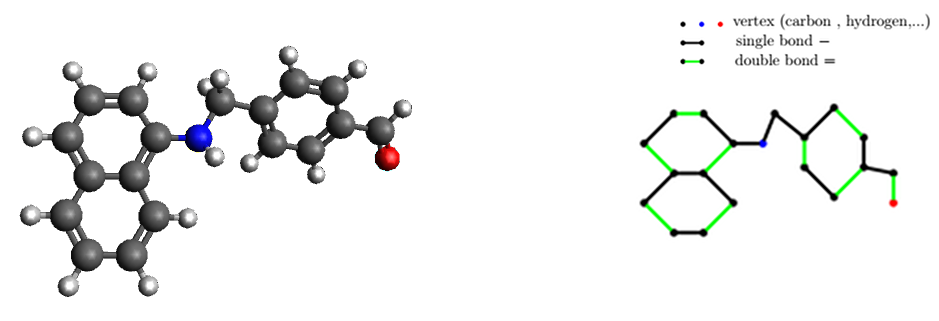}
\caption{Example of a molecule and its corresponding molecular graph.}
\label{fig:graph}
\end{center}
\end{figure}

\subsection{Preliminary }

We consider a simple and undirected labeled graph $G=(V,E)$ with $n = |V|$ the number of vertices and $m = |E|$ for the number of edges in $E = \{e_1, e_2, ..., e_m\}$.

An elementary cycle $c$ can be represented by a vector $v_c = (e_1^c,e_2^c,...,e_m^c)$ where $e_i^c =1$ iff the edge $e_i^c$ belongs to $c$ otherwise $e_i^c= 0$. The\textbf{ \textit{length of a cycle}} $c$ is the number of edges that belongs to the cycle $|c| = \sum e_i^c$.

\begin{definition} 
Let us consider two cycles $c_1$ and $c_2$ of vectors $v_{c_1}  =(e_1^{c_1}, e_2^{c_1},$\ $ ..., e_m^{c_1})$ and $v_{c_2} = (e_1^{c_2},e_2^{c_2},...,e_m^{c_2}) $. The \textbf{union of cycles} $c_1$ and $c_2$  with the boolean operator XOR (symbol $\oplus$) is $c_{12} = c_1 \oplus c_2 $ such that the vector $v_{c_{12}} = (e_1^{c_1} \oplus e_1^{c_2},e_2 ^{c_1}\oplus e_2^{c_2},...,e_m^{c_1}\oplus e_m^{c_2})$.
\end{definition} 

Since  $c_1$ and $c_2$ are elementary cycles, then the union of $c_1$ and $c_2$ is an union of edge-disjoint cycles by definition of $\oplus$.

A $2-$connected component is a maximal (in terms of inclusion) $k-$con\-nected induced subgraph with $k \geq 2$.

\begin{definition}
An \textbf{isthmus} is an edge of $G$ whose deletion increases its number of connected components. An edge is an isthmus if it is not contained in any cycle of $G$.
\end{definition}

 An \textit{isthmus-free graph} is a graph that does not have any isthmus. If a graph $G$ has $p$ isthmus then its number of $2-$connected components $K$ is such that $K-1 \leq p$; each connected component of a bridgeless graph is $2-$edge-connected. The $2-$connected components in a graph are connected in $G$ by  isthmus-chains (a chain which all edges are isthmus).

\section{Cycle structure of a molecular graph}

 In this section, the goal is to define a molecular representation which encodes the interconnection between the cyclic parts of the molecule. We assume that the cyclic part ( $k-$connected component with $k \geq 2$) describes the structure and the acyclic part describes chemical properties of the molecule. So, similar molecules may certainly have similar cyclic parts. 

This cyclic structure of molecular graph is based on the interconnection of its induced cycles. However, we do not compute and represent in the graph of cycles all the elementary cycles as there can be an exponential number of such cycles. In order to get a compact representation of the molecule cycles, we can use minimum cycle bases\cite{Horton1987} of the graphs.

\begin{definition}
A generator $\zeta$ is a set of cycles such that for each cycle $c$ of $G$ there is a set of cycles $c_1,c_2,...,c_k$ in $ \zeta $ such that $c = c_1 \oplus c_2 \oplus  ... \oplus c_k$.
\end{definition}

The weight of a generator is the sum of the lengths of its cycles. We denote $\zeta^i$ the generator of cycles with weight equal to or lower than $i$.

\begin{definition} 
A cycle basis of a graph is a minimal generator in terms of inclusion. 
\end{definition}

A minimum cycle basis is a cycle basis with a minimum weight. Note that, for a graph  we can have more than one minimum cycle basis. It is difficult to choose a canonical cycle basis to represent the interconnection of cycles because of the non uniqueness of cycle basis in a graph (see Figure \ref{fig:basis}). This means that depending on the choosen algorithm to compute a minimum cycle basis and the vertex labelling, two isomorphic graphs may have different cycle basis as results. Thus, we cannot only refer to one cycle basis to decide on the similarity between molecules.

\begin{figure}[H]
\begin{center}
\includegraphics[scale=0.67]{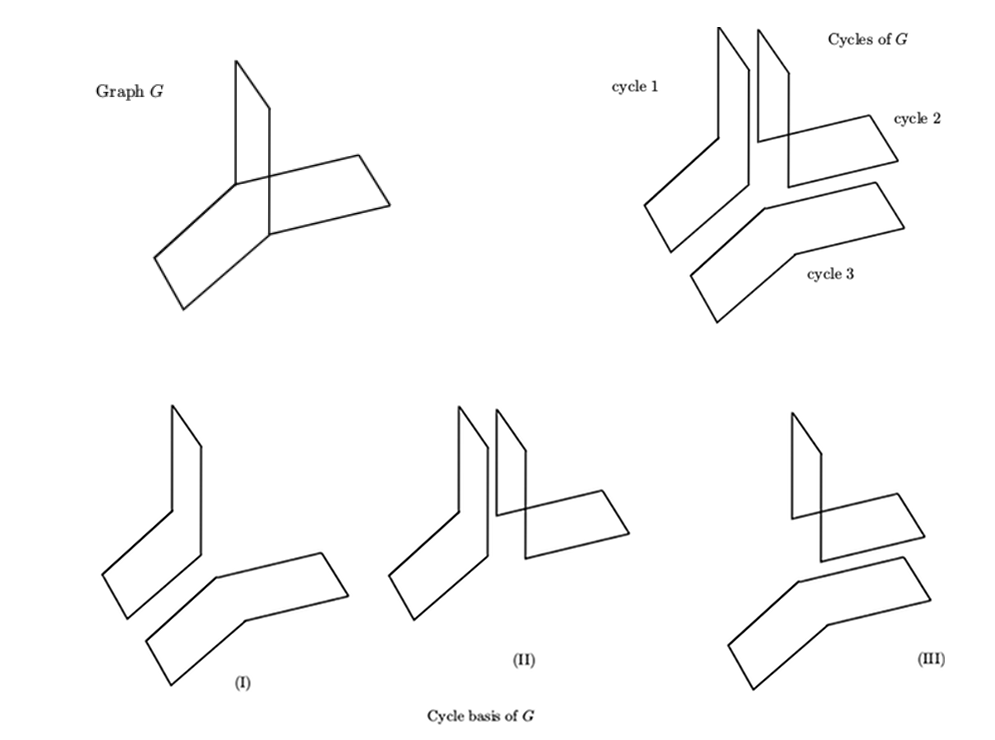}
\end{center}
\caption{Differents cycle basis}
\label{fig:basis}
\end{figure}

Vismara \cite{Vismara1997} reported that the union of minimum cycle basis of the graph is a canonical generator. In the same paper, it is said that the union of minimum cycle basis is the smallest canonical set of cycles which computes the cyclic structure of a graph and the number of cycles of the union of cycle basis can be exponential. Referring to the definition of a cycle basis, the union of minimum cycle basis in a graph is a generator. Although the algorithm proposed by Vismara computes a compact representation of the potentially exponential-sized set, there is no algorithm to list all the cycles of the union of minimum cycle basis. In the following section we will introduce an algorithm to compute a canonical generator of a molecular graph.

Given a molecular graph $G$ and a canonical generator $\zeta$ of $G$, our goal is to compute a graph $G^\zeta$ representing the cyclic part of $G$ and describing the interconnection between cycles of the molecular graph $G$.

\subsection{Graph of cycles for molecular graph}

Before defining formally the graph of cycles, we illustrate and explain it on one example.

\begin{example}
\normalfont

Let us consider the molecular graph of quinine, with $\{c_1, c_2, c_3,$ \ $c_4, c_5\}$ a canonical generator containing $5$ cycles (see Figure \ref{quinine}). These cycles are the vertices of the corresponding graph of cycle. In terms of similarity between molecules, when considering interaction between cycles in a molecular graph, its is important to distinguish cycles sharing some vertices (like cycles $c_1$ and $c_2$) and cycles linked by a path (like $c_2$ and $c_3$). It is why we consider two types of edges in the graph of cycle of a molecule. Firstly, the type $1$ is used for closed cycles \textit{i.e.} for cycles sharing at least one vertex in the molecular graph. Each edge of type $1$ has as label value the number of shared bonds. For instance, the plain blue edges on Figure \ref{quinine} are of type $1$. The edge between $c_1$ and $c_2$ is equals to $1$ because they have one bond in common. Secondly, the type $2$ is for cycles with a relationship than can be easily broken (more often the cycles are not closed in the molecular graph). Edges of type $2$ have as label value the length of a shortest path between the corresponding cycles in the molecular graph. For example, the dashed green edges on Figure \ref{quinine} are of type $2$.

\begin{figure}[H]
\includegraphics[scale=0.2]{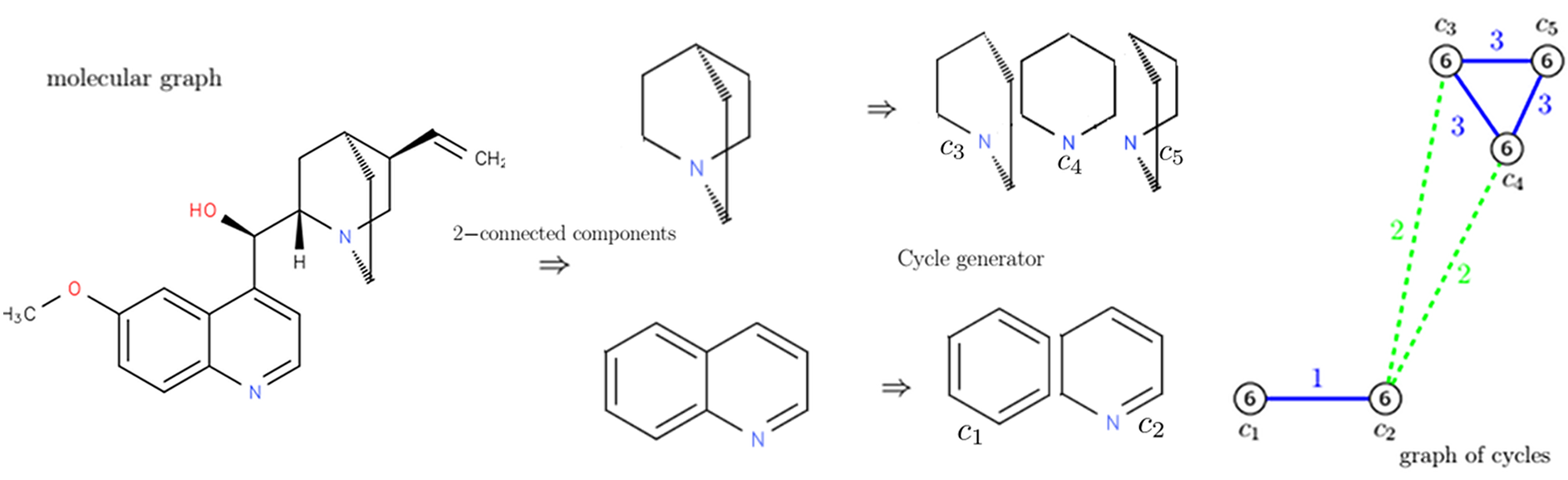}
\caption{Molecular graph, $2-$connected components and cycles of the generator.}
\label{quinine}
\end{figure}
\end{example}

In terms of similarity measurement, we will also have to upper bound the size of the considered cycles to be considered in the target molecular graph. Let us focus on two molecules considered as structurally similar : strychnine and vomicine. Indeed, as it is illustrated in Figure \ref{strychnine}, if we consider all the sizes of cycles in the vomicine, the two molecular graphs appear to be not similar. But, if we do not consider the cycles of size $9$ in the vomicine molecular graph, then the two obtained graphs of cycles are similar. In fact these cycles of size $9$ aren't cycles involved in the structure of the molecule but rather a connection between the structural part of the molecule and an azote atom. When in this case, reducing the graphs of cycles to cycles with size lower than or equal to $7$ is relevant, and it will be the case in most cases. It is why we introduce parameter $j$ in the next definition in order to allow or remove cycles for similarity.

\begin{figure}[H]
\includegraphics[scale=0.24]{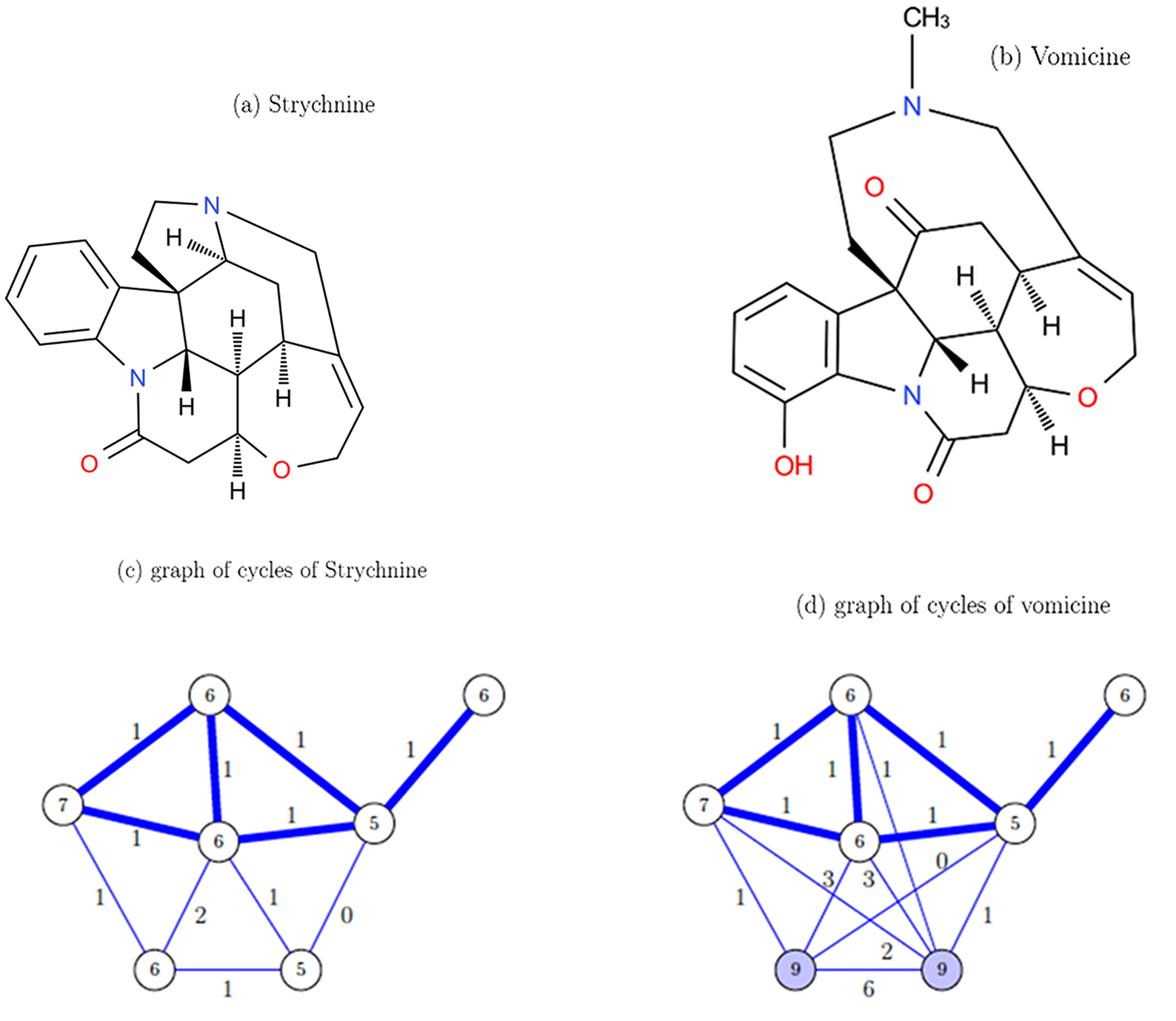}
\caption{Similar molecules :  Strychnine and Vomicine with their graph of cycles}
\label{strychnine}
\end{figure}

\begin{definition}
Let us consider a generator $\zeta$ and an integer $j$. The generator $\zeta$ is $j-$hierar\-chical if the subset of cycles of $\zeta$ with length equal or lower than $j$ can generate all the cycles of length lower than or equal to $j$ in $G$. 
\end{definition}

We denote by $\zeta_j$ the $j-$hierarchical set of $\zeta$. A generator $\zeta$ is hierarchical iff $\zeta_j$ is $j-$hierarchical for every $j$. 

\begin{lemme}
A minimum cycle basis of any graph is hierarchical.
\end{lemme}

\begin{proof}
\normalfont
Let us consider a minimum cycle basis $B$. Assume that $B$ is not hierarchical \textit{i.e.} there is an integer $j$ such $B_j$ is not $j-$hierarchical. 

Since $B_j$ is not $j-$hierarchical, then there is a cycle $c$ of length lower than or equal to $j$ which cannot be generated with $B_j$. Therefore the cycle $c$ doesn't belongs to $B$. 

Since $B_j$ is a cycle basis, there is a set of cycles $\{c_1, c_2, ..., c_{\alpha}\}$ in $B$ with $ c = c_1 \oplus c_2 \oplus ... \oplus c_{\alpha -1 } \oplus c_{\alpha}$. Let us assume that $c_{\alpha}$ is a cycle of maximum length in the set \{$c_1, c_2, ..., c_{\alpha}$\}. Since $B_j$ doesn't generate $c$ then the size of $c_{\alpha}$ is greater than $j$.

The binary operator $\oplus$ is commutative so $c_1 \oplus c_2 \oplus ... \oplus c_{\alpha -1 }\oplus c = c_{\alpha}$.
We denote by $B'$ the set of cycles obtain by removing $c_{\alpha}$ and adding $c$ in $B$ (\textit{i.e.} $B' = B \backslash \{c_{\alpha}\}\cup \{c\} $). As $\{c_1, c_2, ..., c_{\alpha-1},c\} \subset B'$, $ c_{\alpha} = c_1 \oplus c_2 \oplus ... \oplus c_{\alpha -1 } \oplus c $ and $B$ a cycle basis, so is $B'$. The weight of the cycle basis $B'$ is $|B'| = |B| -|c_{\alpha}|+ | c| $. The weight of $B'$ is lower than the weight of $B$ (a contradiction because $B$ is a minimum cycle basis). Then $B$ is hierarchical.

\end{proof}

\begin{definition}
\label{graphcyclesdefinition}
Let $G$ be  a molecular graph, an integer $j$ and $\zeta_j$  be a $j-$hierar\-chical generator of cycles in $G$. The graph of cycles of $G$ induced by $\zeta_j$ is denoted $G^{\zeta_j} = (V^{\zeta_j},E^{\zeta_j},\mu,\nu,\theta)$ with the edge-set $E^{\zeta_j} = E^{\zeta_j}_{1} \cup E^{\zeta_j}_{2}$ .

 \begin{itemize}
 
 \item The vertex-set $V^{\zeta_j}$ is $\zeta_j$.
 \item The edge-set $E^{\zeta_j}$ define the relationship between cycles of  $V^{\zeta_j}$ according to to their proxi\-mity in $G$.
 \begin{itemize}
 \item $[c_1,c_2] \in E^{\zeta_j}_{1}$ iff $c_1$ and $c_2$ belong to the same $2-$connected components of $G$ and they have at least one common vertex.
 \item $[c_1,c_2] \in E^{\zeta_j}_{2}$ iff  $c_1$ and $c_2$ belong to different  $2-$connected components and there is a path $p$ from a vertex of $c_1$ to a vertex of $c_2$ in $G$ such that all edges of $p$ doesn't belongs to a cycle in $V^{\zeta_j}$.
 \end{itemize}
 \item For each vertex $c \in V^{\zeta_j}, \mu(c)$ is the weight of the cycle $c$;
 \item For each edge $e \in E^{\zeta_j}_k,$ $\nu(e) = k$;
 \item For each edge $e= [c_1,c_2]  \in E^{\zeta_j}$, $\theta(e)$ is the distance from $c_1$ to $c_2$ in $G$.  If $e \in E^{\zeta_j}_{1}$ then $\theta(e)$ is the number of common edges between $c_1$ and $c_2$ in $G$. Otherwise $\theta(e)$ is the length of the shortest path between a vertex of $c_1$ and a vertex of $c_2$ in $G$.
 \end{itemize}
\end{definition}

In the Example \ref{quinine}, we have $\mu(c_1) = 6$ as the length of the cycle $c_1$, $\nu([c_1,c_2]) = 1 ,$ $ \nu([c_2,c_3] = 2$ and $\theta([c_2,c_3]) = 2$ (the smallest path from a vertex of $c_2$ to a vertex of $c_3$ in the molecular graph).
In the following section, we describe how to compute a generator of cycles for a molecular graph and the relationship between its cycles.

\subsection{Cycles generator of a graph}
\label{edgecover}

A generator $\zeta$ of cycles as we define contains cycles such that each edge which belongs at least to a cycle is represented. It is computed by using a minimum cycle basis and adding additional cycles.

In this section, we present algorithms to compute a cycle generator $\zeta_j$. Let us consider a molecular graph $G=(V,E,w_V,w_E)$ that may be non connected. We called the structural graph of a molecular graph to be the maximum subgraph of $G$ without any vertex with a degree lower than $2$ in the subgraph. We delete the bridges in the structural graph and we denote $G_i$ with $ i \in [1..K]$ the $2-$connected components ($K$ is the number of components) computed with the bridgeless graph of the structural molecular graph.
\begin{algorithm}[H]
\label{algo1}
\SetAlgoLined
\KwData{A molecular graph $G$, an integer $j$ = maximum length of cycles .}

\KwResult{Generator $\zeta_j$ of a molecular graph $G$}
$\zeta_j= \emptyset $ set of cycles of the generator\;
$T = \emptyset $ set of cycles\;
Remove all the bridges and leaves in $G$ \;
Extract the $2-$connected components ( $G_1,G_2,...,G_K$ components)\;
Compute for each $G_i, i \in [1..K]$, a minimum cycle basis $B^i =\{c^i_1,c_2^i,...,c_k^i  \}$  \;
 \ForEach{Cycle basis $B^i$} 
 {
 	\ForEach{Couple of cycles $ c_a^i , c_b^i \in B^i$}
 	{
 	Let  a cycle $ c =  c_a^i \oplus c_b^i$\;
 	
 	\If{$c$ is an elementary cycle in $G_i$ and $c \notin B^i$ and $|c| = max(|c_a^i|,|c_b^i|)$}
 	{
 	Add $c$ to $T$\;
 	}
 	}
 	
 	$B^i = B^i \cup T$ \;
 	$T = \emptyset $ \;
 }

 $\zeta_j = \cup B^i$ \;
 Remove in $\zeta_j$ all the cycles with a length bigger than the parameter $j$\;

\caption{Generator of a graph.}

\end{algorithm}
\vspace{0.5cm}

We choose to compute a minimum cycle basis of a graph with the Horton algorithm \cite{Horton1987}. The algorithm is described here : 

\vspace{0.3cm}
\begin{algorithm}[H]
\label{algo2}
\SetAlgoLined
\KwData{A graph $G$}
\KwResult{A minimum cycle basis of $G$}
$B= \emptyset $\;
 Find shortest chains between all pairs  of vertices in each $G$\;
 \ForEach{ vertex $v$ and edge $[x,y]$ in each $G$}
	{Create the circuit $C(v,x,y) = P( v,x) + P(v,y) +[ x, y]$\;
	\If{$P( v,x) \bigcap P(v,y) = \{v\}$}
	{ Add $C(v,x,y)$ to $B$\;
	}
	}
 Order all the cycles of  $B$ by length\;
 Use a greedy algorithm (Gauss elimination) to find the minimum cycle basis $B$ from its set of cycles\;

\caption{Horton algorithm.}

\end{algorithm}
Horton algorithm is polynomial  $O(n*m^{3})$ \cite{Horton1987}. The complexity of the algorithm \ref{algo1} is lower than $O(n^2*m^{3})$:

\begin{itemize}
\item \textbf{ Step 2:} is polynomial $O(m^{2})$
\item \textbf{Step 4:} Horton algorithm is called $K$ times and each $G_i$ have at least $n$ vertex. An upper bound is $O(K*(n*m^{3}))$
\item \textbf{Step 5 to 12: } number of operations is : $K*k_i^2$, with $k_i$ the number of cycles in the computed minimum cycle basis of $G_i$.
\end{itemize}

Since similar molecules have similar structural parts, to compute the similarity between molecules we are going to do it on their graphs of cycles. In the next section we present a similarity calculation on graphs and we check if the results are coherent in terms of similarity of molecules.

\section{Similarity calculation and experimental results}

\subsection{Similarity calculation}
To measure similarity of two molecules on their corresponding graphs of cycles, we solve the Maximum Common Edge Subgraph (MCES) problem \cite{Raymond2002}. Considering two graphs $G=(V,E)$ and $G'=(V',E')$, a mapping of $G$ on $G'$ is a function  $\pi:V \rightarrow {\cal P}(V')$, such that for any $v \in V$, $\pi(v)\neq \emptyset$. We say that $G$ is \mbox{$\pi$-isomorphic} to $G'$ iff there exists an isomorphism $\cal I$ between $G$ and $G'$ such that, for any $v\in V$, we have ${\cal I}(v) \in \pi(v)$. Solving problem MCES  constrained by $\pi$ consists in finding the maximum subgraph of $G$ being $\pi$-isomorphic to a subgraph of $G'$. This problem has been shown to be NP-complete \cite{Raymond2002}.

Let us consider two molecules $M_1$ and $M_2$ and their corresponding graphs of cycles $G^{\zeta_j}_1=(V^{\zeta_j}_1,E^{\zeta_j}_1)$ and $G^{\zeta_{j'}}_2=(V^{\zeta_{j'}}_2,E^{\zeta_{j'}}_2)$. In our context, mapping $\pi$ is defined such that  for any $v \in V^{\zeta_j}_1 , \pi(v)= \{ v' | v' \in V^{\zeta_{j'}}_2$  \textit{and} $  | |v| -|v'|| \leq 0.2 *min (  |v|,|v'| )\}$. This function $\pi$ in graph of cycles allow two cycles to match if they have similar length. The value $0.2$ has been fixed experimentally. In our experiments, we compare the MCES calculations\cite{Raymond2002} on the molecular graphs and on the graph of cycles. For each computation we use a distinct isomorphism $\pi$. For two molecular graphs, the function $\pi$ maps maps atoms of the same type. Now consider two graph of cycles $G_{\zeta_1}$, $G_{\zeta_2}$ and  $G_{\zeta_{12}} = (V^{\zeta_{j''}}_{12},E^{\zeta_{j''}}_{12},\mu,\nu,\theta)$ a Maximum Common Edge Subgraph (MCES) of $G_{\zeta_1}$ and $G_{\zeta_2}$ constrained by $\pi$. We remind that the function $\mu$ indicates the length of each cycle; the function $\nu$  indicates the relation between each pair of  connected cycles (if they share vertices or not) and the function $\theta$ gives the label of edges between cycles (see Definition \ref{graphcyclesdefinition}). 

Considering two graphs $G_{\zeta_1}$ and $G_{\zeta_2}$,  the similarity  is the ratio between the sum of vertices and edges of the MCES $G_{\zeta_{12}}$ and the product of the sum of vertices and edges of $G_{\zeta_1}$ and the sum of vertices and edges of $G_{\zeta_2}$ :
 \begin{equation}
    sim(G_{\zeta_1},G_{\zeta_2}) = \frac{(|V^{\zeta_{j''}}_{12}| + |E^{\zeta_{j''}}_{12}| ) ^2 }{(|V^{\zeta_j}_1| + |E^{\zeta_j}_1| )\times (|V^{\zeta_j}_2| + |E^{\zeta_j}_2| )}
 \end{equation}
 Note that finding a maximum common edge subgraph between $G$ and $G'$ is similar to find the maximum clique in the product graph of the linegraph induced by $G$ and $G'$. In the next section, we experiment the similarity on graph of cycles.

 \subsection{Experimental results}
 
Considering some specific molecules in a database, each one to be compared with all the other ones, we evaluate the performances of measuring molecular similarity by using the method given in \cite{Raymond2002}, on the one hand on molecular graphs (MG) and on the other hand on graphs of cycles (GC). Considering $\pi$-isomorphism of molecular graphs, function $\pi$ concerns the type of atoms.

The target database of molecules is a freely available dictionary of small molecular entities called Chemical Entities of Biological Interest ChEBI. This database contains $ 90130$ molecules. According to their structural configuration, we choosed seven molecules in ChEBI. Considering the structural part of molecules from a chemical point of view, some of these molecules have many similar molecules in the database while others don't. For each one, we compute the similarities with all the molecules in the database. We then compare the distributions of the obtain similarities in the two contexts MG and CG, and we focus on the $20$ most similar pairs of molecules in each context. he two contexts MG and CG, and we focus on the $20$ most similar pairs of molecules in each context.  

To make sure that the two methods computed the similarity on the structural part, we removed all the leaves and bridges in all molecular graphs.
 
The computation has been done on a cluster Intel(R) Xeon CPU E$5$-$260$ v$3$ $@ 2.40$GHz with $64$G of RAM. To find a maximum clique in a graph to solve MCES, we did a linear program resolved by SCIP\footnote{http://scip.zib.de/} (Solving Constraint Integer Programs).  Because of the number of molecules in the database ($90 130$ molecules, knowing that many other database are larger) and since the similarity calculation between two graphs may have an exponential runtime due to the NP-completeness of the problem MCES, we chosed to fix an upper bound of similarity computation of similarity for each pair of molecules. This time depend on the size of the considered graphs (MG or GC). For example, if the maximum time for each similarity is $20$ seconds, then the whole computation requires $\pm20$ days on the cluster. As a consequence of the time limitation, some similarities are not computed for some pairs of molecules in the MG context. 

We then compared the distributions of the obtained similarities in the two contexts MG and CG, and we focused on the $20$ most similar pairs of molecules in each context. Our goal is to evaluate and compare the performances of the two approaches  MG and CG  from three points of view : the execution time required to calculate the measure of similarity for each pair of molecules, the capacity of each approach to distinguish the pairs of real similar molecules (\textit{i.e.}, the ones having similar core structures) and finally the capacity of discriminating real similar, meduim similar and not similar pairs of molecules. 

We are going to present the results the set of seven molecules : Docetaxel Anhydrous, Amphotericin B, Strychnine, Quinine, Cholesterol, Manzamine A and Brevetoxin A.
\subsubsection{Docetaxel Anhydrous}
 
Docetaxel Anhydrous has a generator of cycles with different lengths ($4$, $6$ and $8$). The graph of cycles has $6$ vertices and we can see that it maximum connected subgraph with edges of type $1$  is the kernel of this molecule (see Figure $6$). 
 
\begin{figure}[H]
 \label{docetaxel}
 \begin{center}
\includegraphics[width=5.3in]{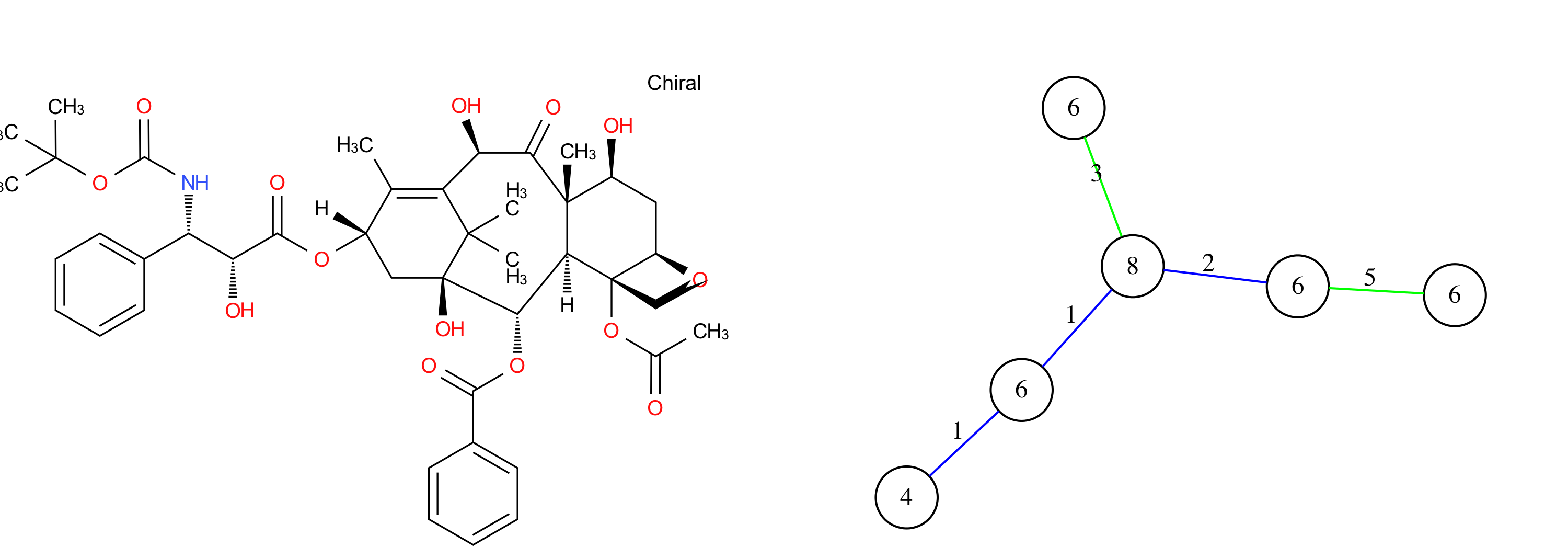}
\end{center}
\caption{Molecular graph and graph of cycles of Docetaxel anhydrous. }
 \end{figure}
 
  Here are the distributions of similarity on MG and GC :
 \begin{figure}[H]
 \begin{center}
\includegraphics[width=2.5in]{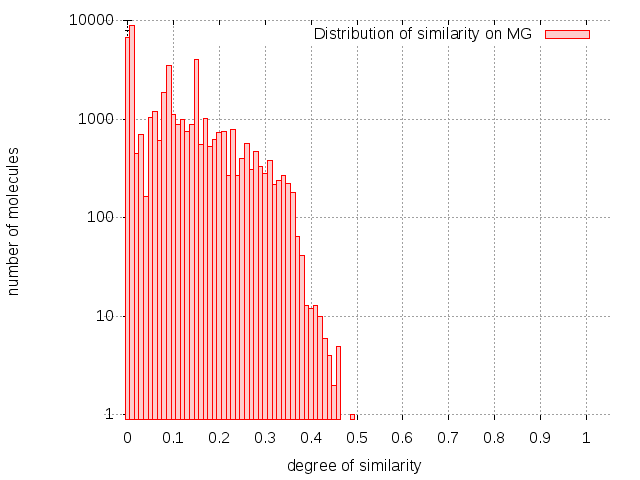}  
 \includegraphics[width=2.5in]{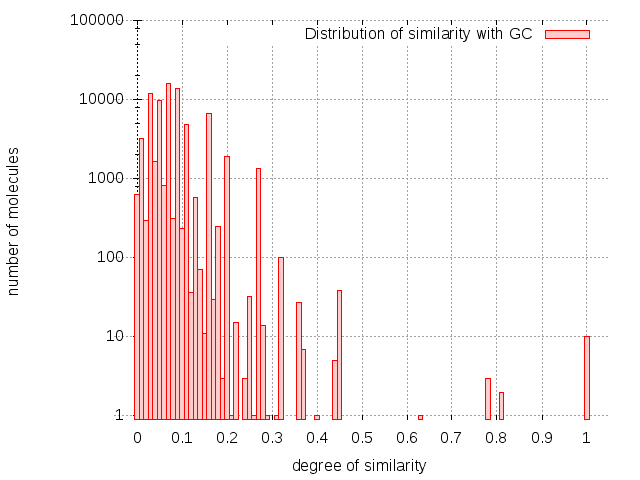} 
\end{center}
\caption{Distribution of similarity of Docetaxel anhydrous on molecular graphs (MG)  and on graphs of cycles (GC).}
\label{doce-dis}
 \end{figure}
 
 According to the distribution of similarity on GC, $4$ categories of similar molecules can be extracted :
 
 \begin{itemize}
 
 \item $9$ molecules are totally similar to Docetaxel (they are isomers). In fact, they have exactly the same graph of cycles. 
 \item $5$ molecules are partially similar; $2$ of them have a similarity degree equals to $0.81$ differ from Docetaxel only on $1-$connected part in MG. Their GC are subgraphs of the graph of cycles of Docetaxel, one cycle linked with an edge of type $2$ is missing.  The $3$ other molecules (with a degree of similarity of $0.78$) have the same structure as Docetaxel with more cycles. The GC of these molecules have GC of Docetaxel as subgraph of (they have one cycle more and two edges of type $2$).
 \item $1$ molecule is the kernel of Docetaxel. The degree of similarity is $0.63$).
 \item The rest of molecules with a degree lower than $0.45$ are not similar to the target molecule. 
\end{itemize}  
 
 In the distribution of similarity on MG, we fixed $30$ seconds to compute the similarity of two molecules. Over $46 846$ of $90 130$ molecules where not computed (about  $51.9 \%$). None of the molecules in top $20$ are chemically similar to to Docetaxel.

 \begin{figure}[H]
 \label{similardoce}
 \begin{center}
\includegraphics[width=5in]{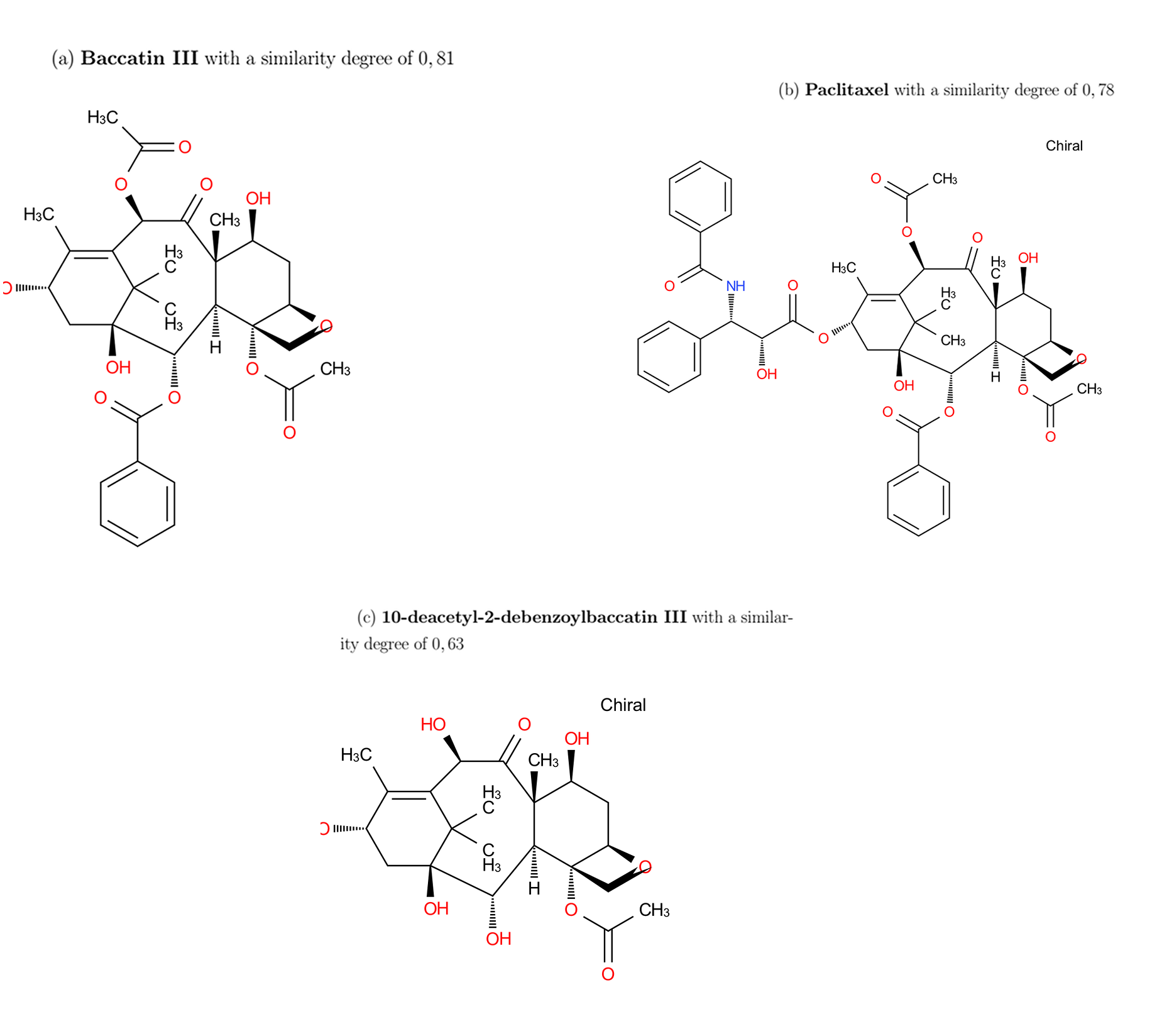}
\end{center}
\caption{ Three similar molecules to docetaxel with GC similarity}
 \end{figure}
 
\subsubsection{Amphotericin B}
 
 Amphotericin B has a particular cyclic structure so it mimimum cycle basis contains $3$ cycles with a particular cycle of length $36$. The corresponding graph of cycles thus contains $3$ vertices  (Figure \ref{AmpB-st}).
 
 \begin{figure}[H]
 \begin{center}
\includegraphics[width=5.3in]{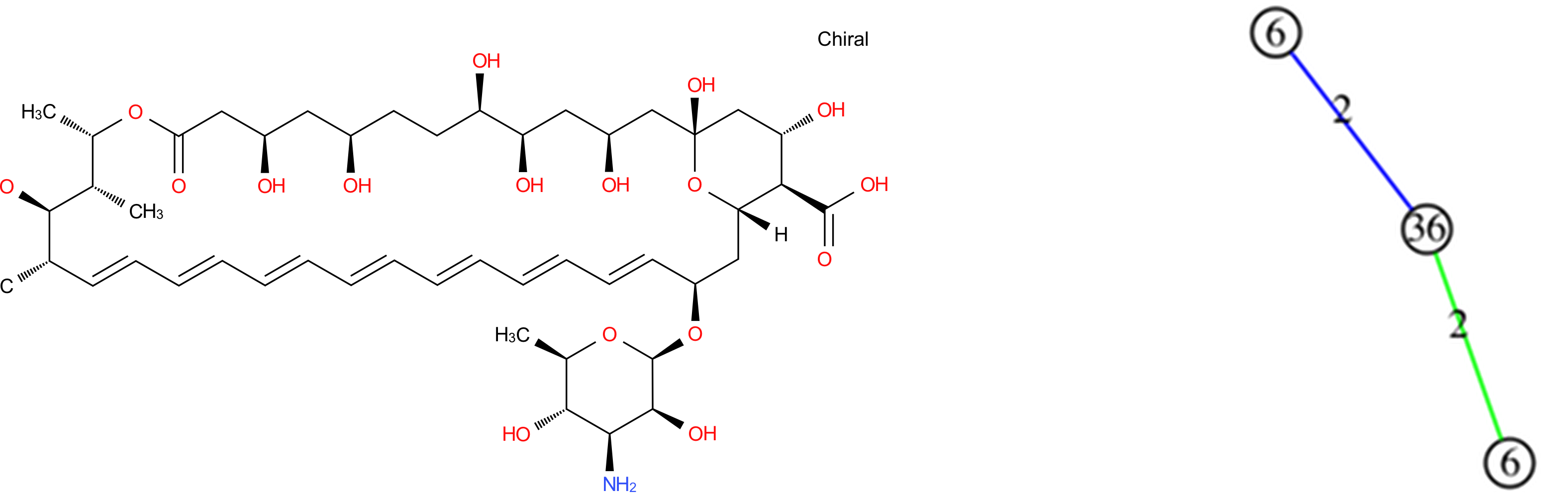}  
\hspace{2cm}
\end{center}
\caption{Amphotericin B molecular graph and it graph of cycles.}
\label{AmpB-st}
 \end{figure}

The GC distribution of similarities concerns all the molecules of the database. This distribution given in Figure \ref{AmpB-dis} shows $11$ molecules fully similar to the target one (degrees of similarity equal to $1$), and another distinguished set of molecules being partially similar to it (degrees of similarity equal to $0.7$ or to $0.6$). The other molecules can be considered as different from the target molecule (similarity lower than $0.5$). Thus, the calculation using cycle graphs clearly discriminates the molecules into three classes, which the molecular graph approach does not do. Moreover, MG approach does not succeed in calculating similarity degrees for several molecules classified as very similar by the GC approach ($50 932$ over $90 130$ molecules where not computed; that is $56.5 \%$). This is due to a too important running time needed; the computation is stop because of the upper bound ($20$ seconds). Indeed, the required computation time is far exceeding the imposed limit.  
\begin{figure}[H]
 \begin{center}
\includegraphics[width=2.5in]{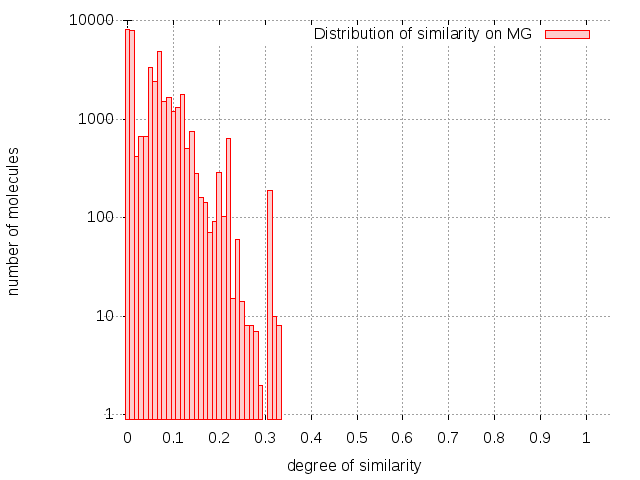}  
\hspace{0.5cm}
 \includegraphics[width=2.5in]{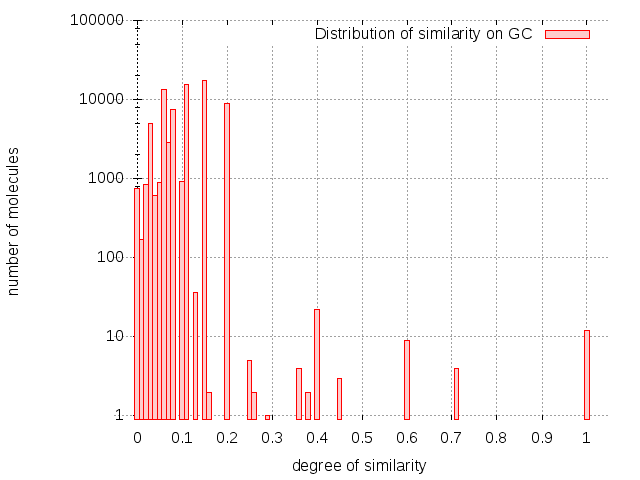} 
\end{center}
\caption{Distribution of similarity of Amphotericin B on molecular graphs (MG) and on graphs of cycles (GC).}
\label{AmpB-dis}
 \end{figure}
 Most of the totally similar molecules provided by the GC approach are either isomers of amphotericin B (amphotericin B methyl ester) or member of the same family (nystatin A1). Amphotericin belongs to the family of antifungal. The other fully similar molecules are not intuitively similar to amphotericin B considering their molecular graphs but the similarity in terms of cycle structure are chemically relevant (Figure \ref{similarampho}). The molecules with degree of similarity equal to $0.7$ in the GC distribution are the ones  such that their graph of cycles have the one of Amphotericin B as subgraph, and the molecules with degree of similarity $0.6$ are the ones which graph of cycles is the subgraph of the one of Amphotericin B. Note that these molecules are not discriminated in the MG approach.

 \begin{figure}[H]
 \begin{center}
\includegraphics[width=5.4in]{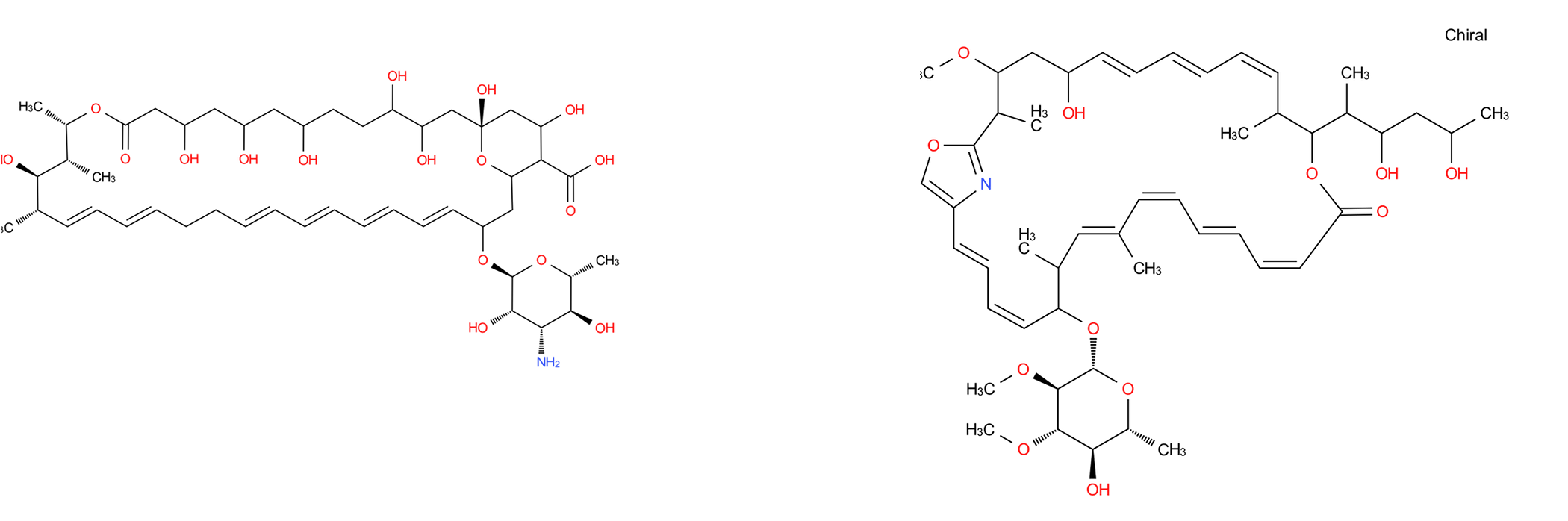}
\end{center}
\caption{nystatin A1 (ChEBI id $ 47 3992$) and Chivosazole A (ChEBI id $80 057$)}
 \label{similarampho}
 \end{figure}
 
\subsubsection{Strychnine}
 
The molecular graph of strychnine is a $2-$connected component with cycles of different length. 
\begin{figure}[H]
 \label{strychninemol}
 \begin{center}
\includegraphics[width=5in]{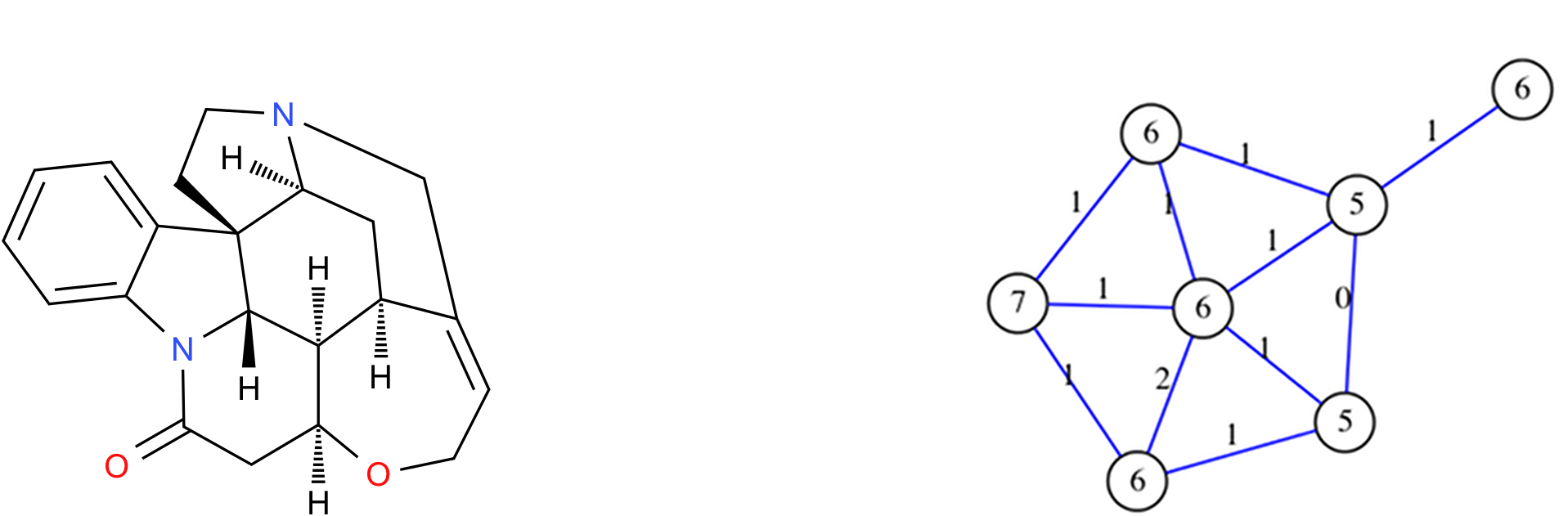}
\end{center}
\caption{Molecular graph and graph of cycles of strychnine. }
 \end{figure}
 
 Over $31 174$ on $90 130$ molecules where not computed for MG ($34.6 \%$). In MG, the six first molecules (with a degree of similarity equals to $0.8$) are the same than the first on GC  (with a degree of similarity equals to $1.0$). All the top $20$ molecules are the same in both methods except when $j=7$, a new molecule appear in GC at position $15$.
 \begin{figure}[H]
 \begin{center}
\includegraphics[width=2.6in]{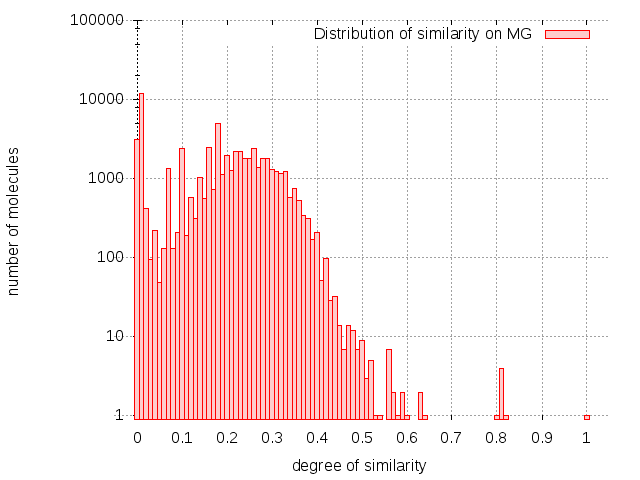}  
\includegraphics[width=2.6in]{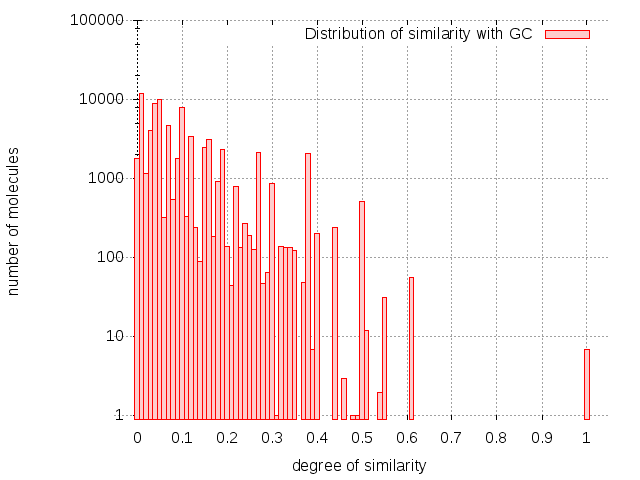} 
\end{center}
\caption{Distribution of similarity of strychnine on molecular graphs (MG) and on graphs of cycles (GC) with $j =9$.}
\label{strychnine-dis}
 \end{figure}

However in GC, the molecule vomicine appears to be similar to strychnine with a degree of similarity equals to $0,32$ (ranking $3720$ over $90 130$) when the parameter $j \geq 9$. Using the same graph of cycles with $j = 7$ (as explain in Figure \ref{strychnine}), the same molecule has a similarity value of $0.68$ with strychnine and a ranking $15 $ over $90 130$. Chemically, these two molecules are similar so it appears important to choose a good value of $j$. 

Here are some molecules similar to strychnine with both GC and MG :

\begin{figure}[H]
 \label{similarstrychnine}
 \begin{center}
\includegraphics[width=2.5in]{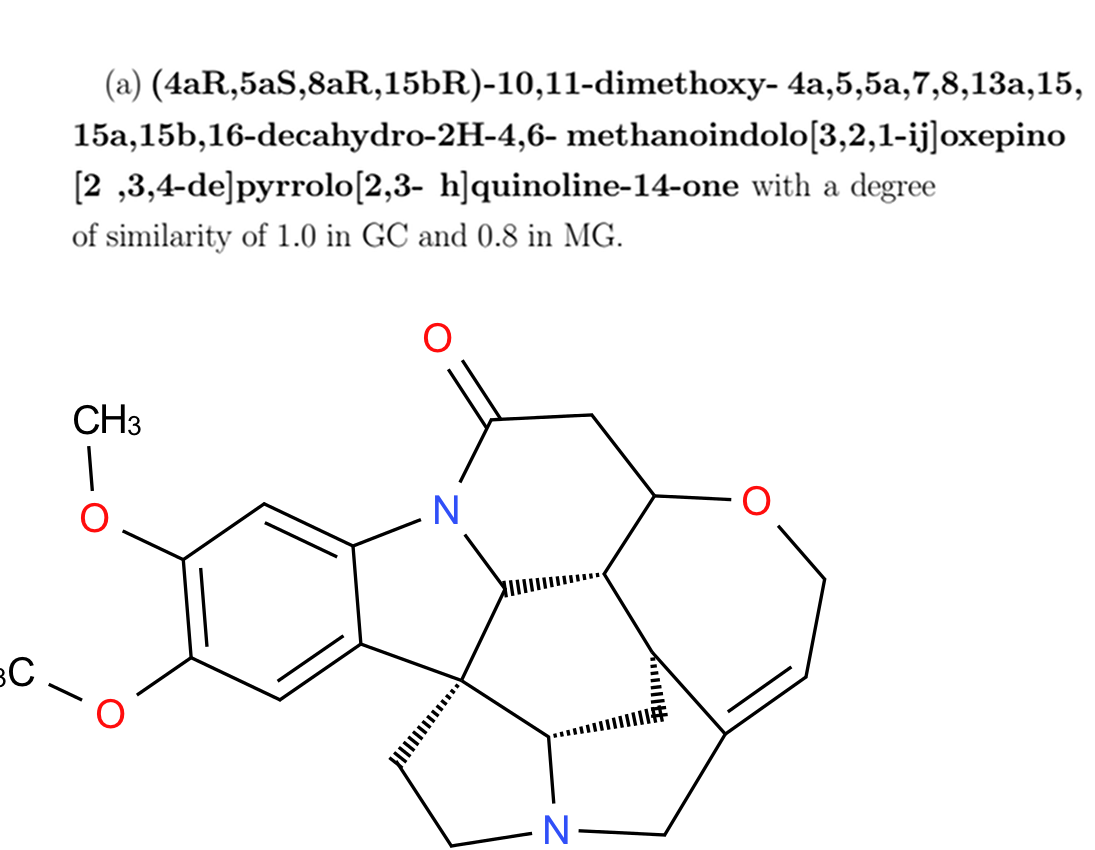}  
\includegraphics[width=2.5in]{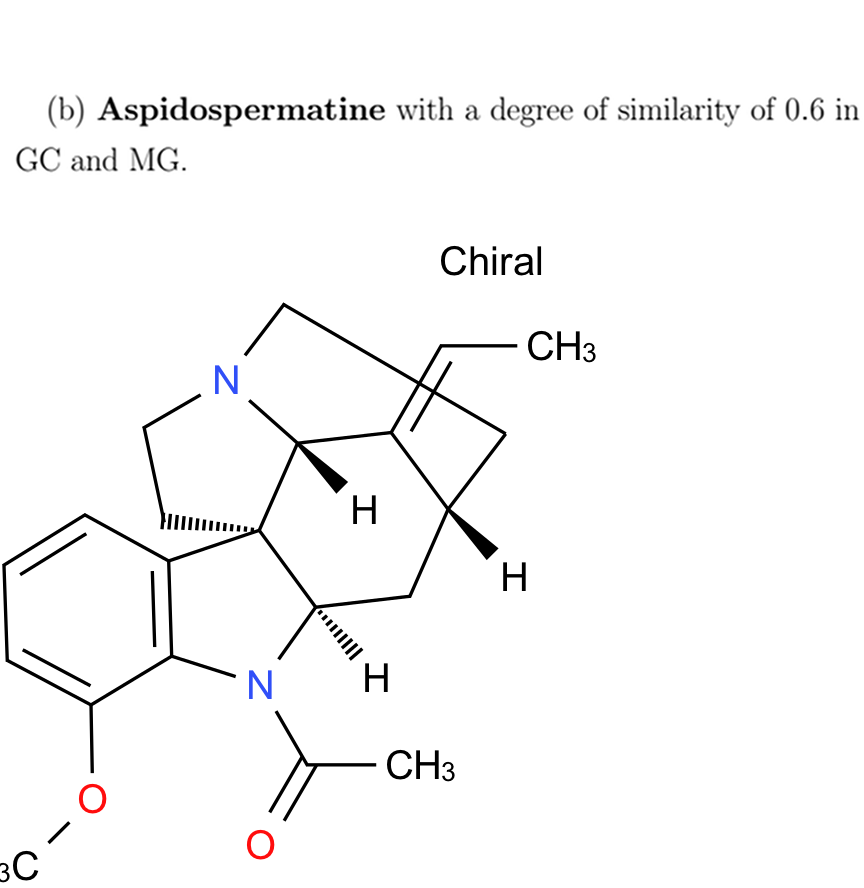} 
\end{center}
\caption{Results of similarity for strychnine with GC and MG }
 \end{figure}

\subsubsection{Quinine}
 
Quinine is a small molecule (with 25 atoms) with a generator of cycles consisted of cycles of length $6$.
 
\begin{figure}[H]
 \label{quininemol}
 \begin{center}
\includegraphics[width=5.4in]{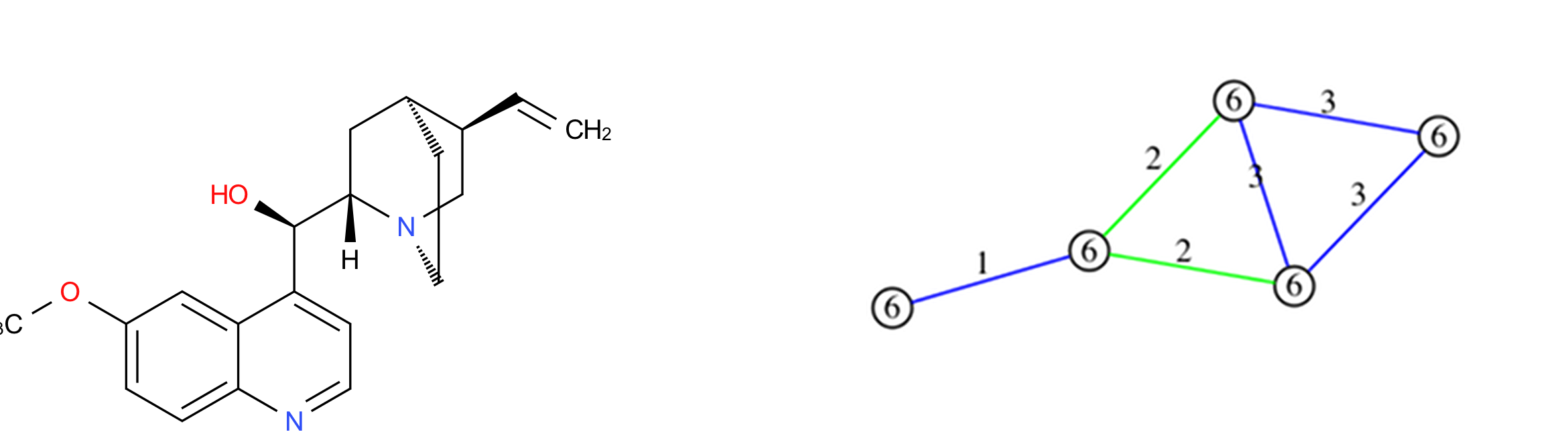}
\end{center}
\caption{Molecular graph and graph of cycles of Quinine. }
 \end{figure}
 
 In GC, there is $24$ molecules with a similarity of degree $1.0$. All of them are members of the same family. The results of similarity doesn't differ isomers (similarity degree of $1.0$) of quinine from others while MG is more precise about that. Molecules with a degree of $1.0$ in MG have exactly the same structural part (same atoms and type of bonds).
 \begin{figure}[H]
 \begin{center}
\includegraphics[width=2.6in]{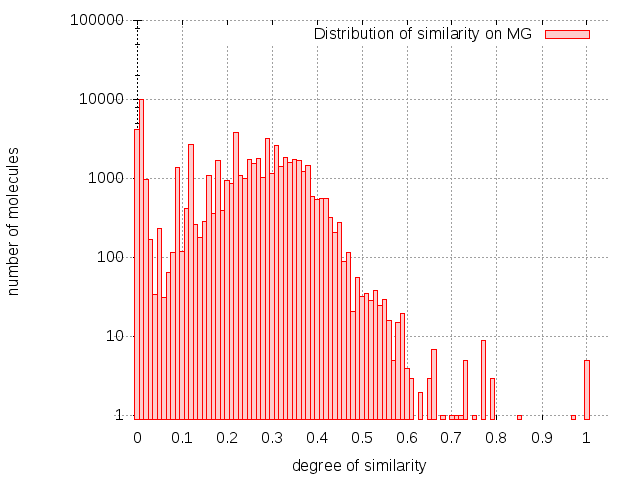}  
 \includegraphics[width=2.6in]{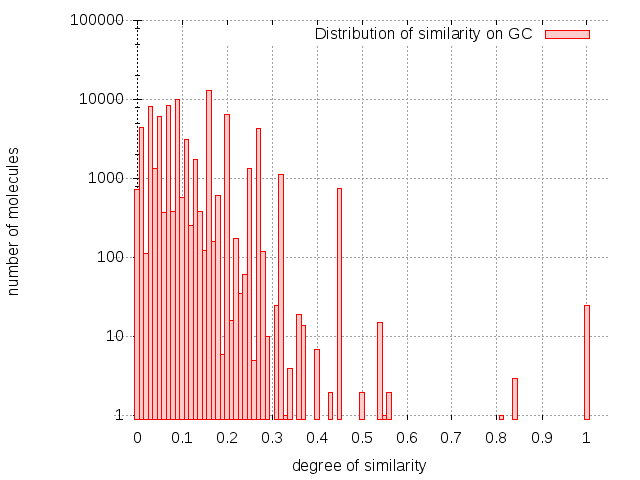} 
\end{center}
\caption{Distribution of similarity of quinine on molecular graphs (MG)  and on graphs of cycles (GC).}
\label{quinine-dis}
 \end{figure}
 
 Molecule $(a)$ in Figure \ref{quinine-dis} is in position $1$ both methods while molecule $(b)$ is n\textsuperscript{o} $1$ in GC with a degree of similarity of $1.0$ but at n\textsuperscript{o} $2197$ with a degree of $0.41$ in MG because of the type of the bonds. Optochin is an isomer of quinine.
 
\begin{figure}[H]
 \label{similarquinine}
 \begin{center}
\includegraphics[width=2.5in]{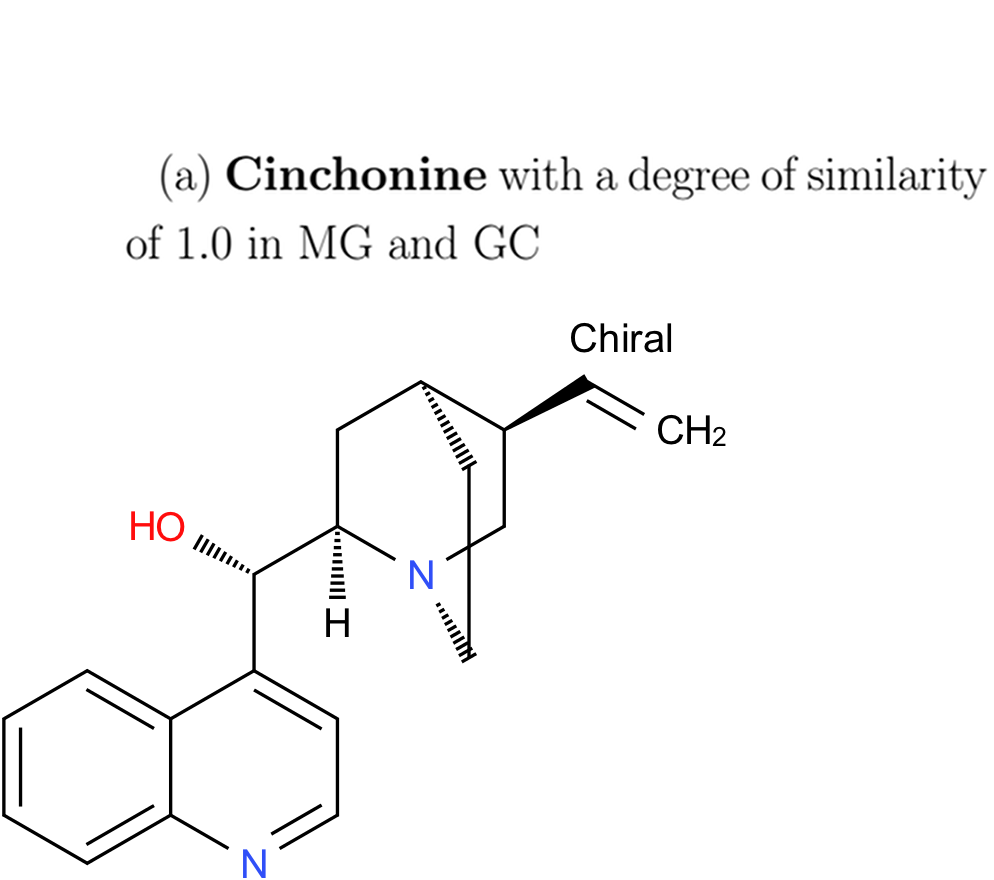}  
\includegraphics[width=2.5in]{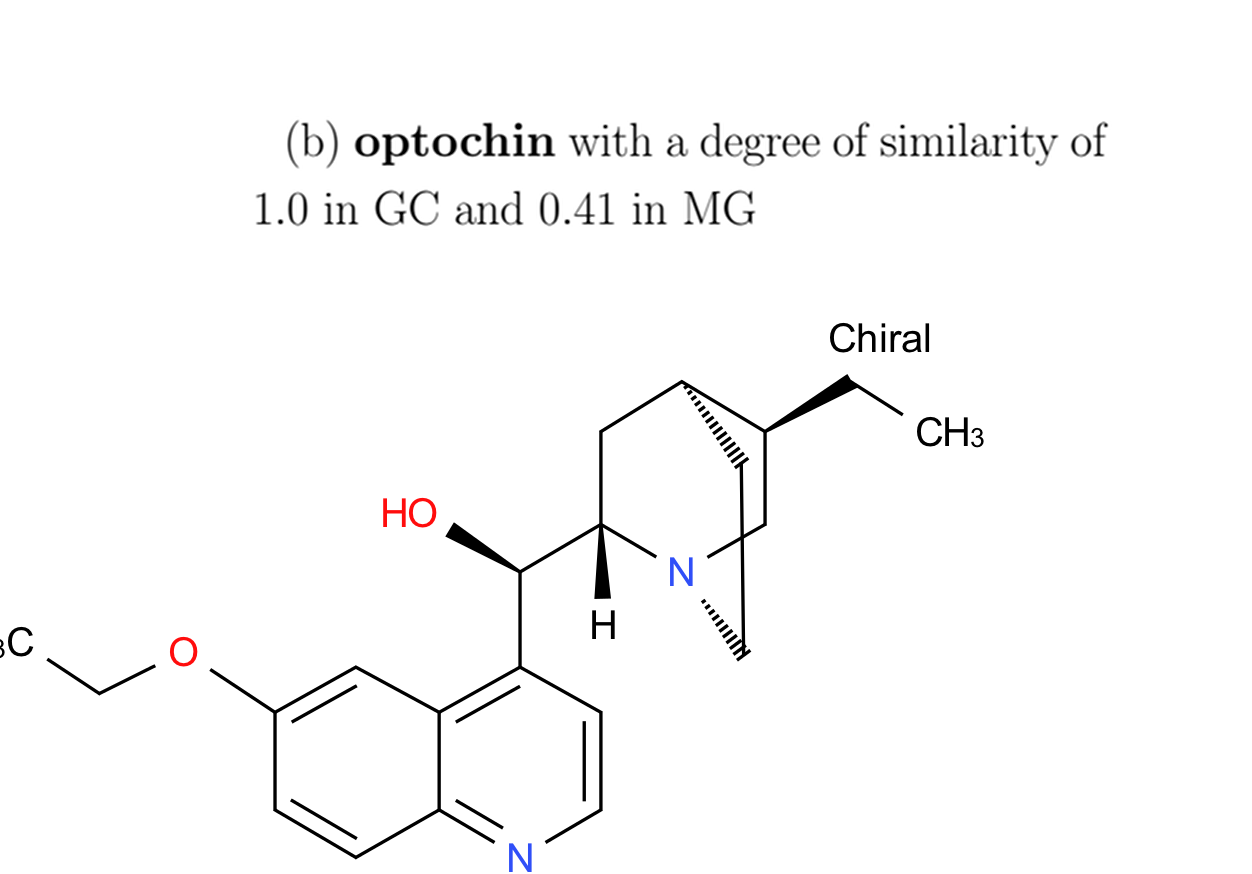} 
\end{center}
\caption{Results of similarity for quinine with GC and MG }
 \end{figure}
Over $29 784$ of $90 130$ molecules where not computed for MG ($33 \%$). Nevertheless, three molecules with a degree of $0.78$ in MG are not similar to quinine and those with a lower degree $0.77$ are similar to quinine. For example the molecular graph Sarpagine (Figure \ref{similarquinine_2}) has as subgraph the molecular graph of quinine. But this subgraph break a cycle (structure) of the molecule.
These results show that finding a maximum common edge subgraph on molecular graph doesn't consider the structural part of the molecule.

\begin{figure}[H]
 \begin{center}
\includegraphics[width=2.6in]{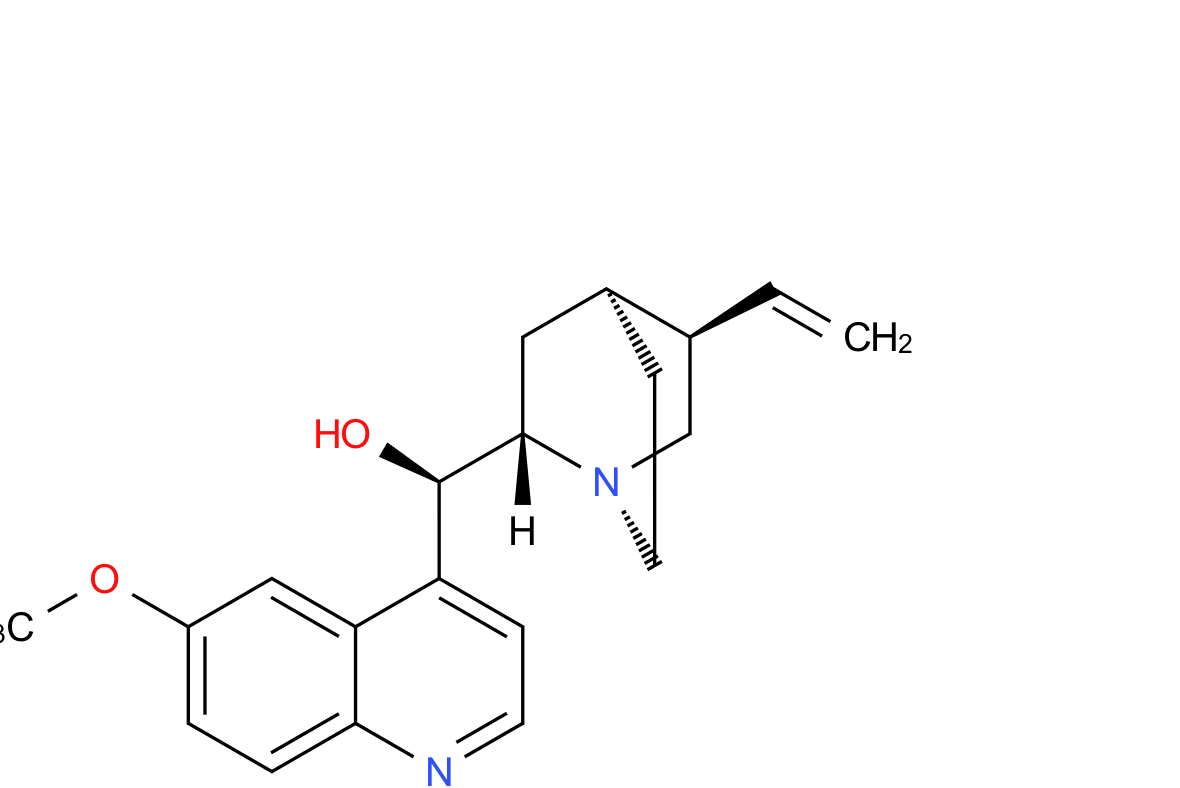}  
\includegraphics[width=2.6in]{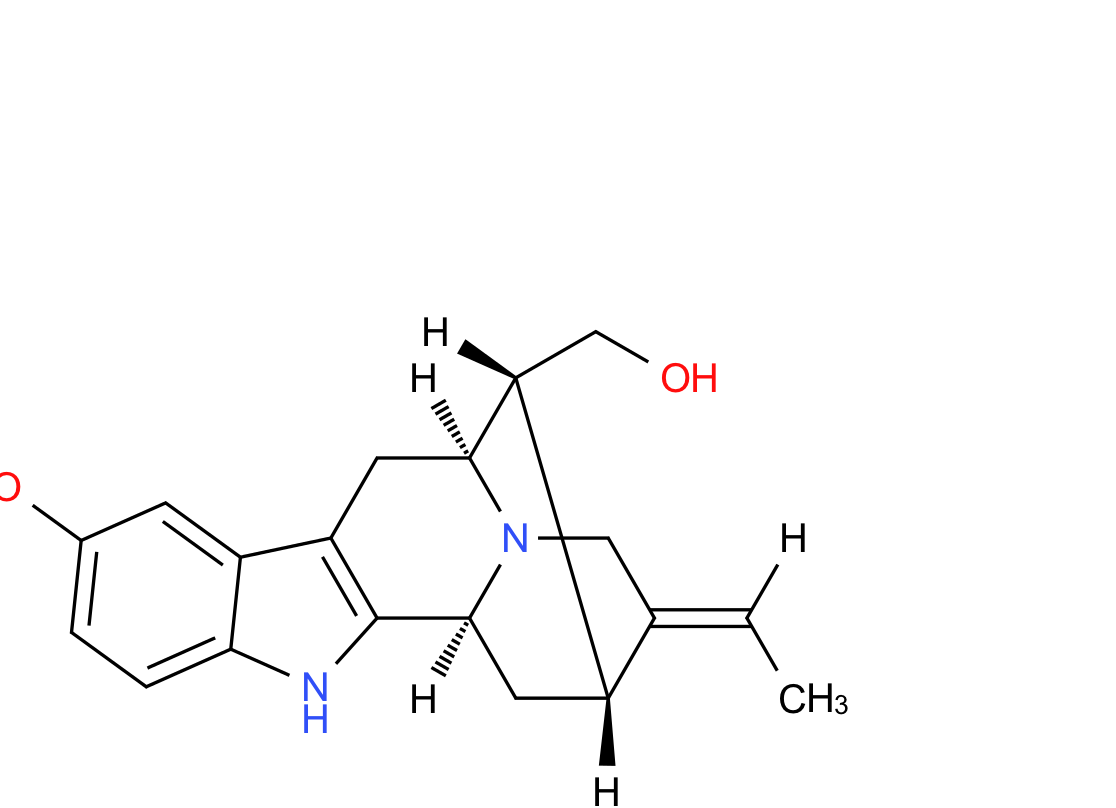} 
\end{center}
\caption{Quinine and Sarpagine similar with $0.78$ in MG }
 \label{similarquinine_2}
 \end{figure}
 With GC, molecules having the same structural part as Sarpagine are ranked with $0.3$ of similarity and they doesn't appear similar to quinine. In fact these molecules are not chemically similar to quinine.
 
\subsubsection{Cholesterol}
 
\begin{figure}[H]
 \label{cholesterol}
 \begin{center}
\includegraphics[width=5.4in]{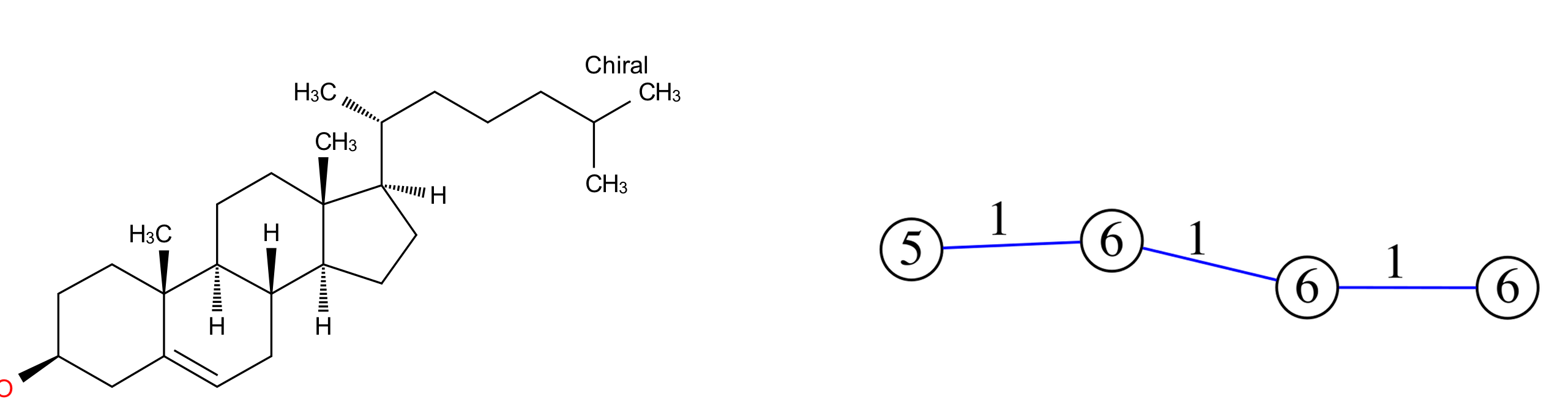}
\end{center}
\caption{Molecular graph and graph of cycles of Cholesterol. }
 \end{figure}
 
The molecule cholesterol is a small molecule having many similar molecules in ChEBI, both methods returns similar results. However the results on MG is more precise in term of type of molecular bonds between atoms.  Over  $31 178$ of $90 130$ molecules where not computed for MG ($41.2 \%$).

 \begin{figure}[H]
 \begin{center}
\includegraphics[width=2.6in]{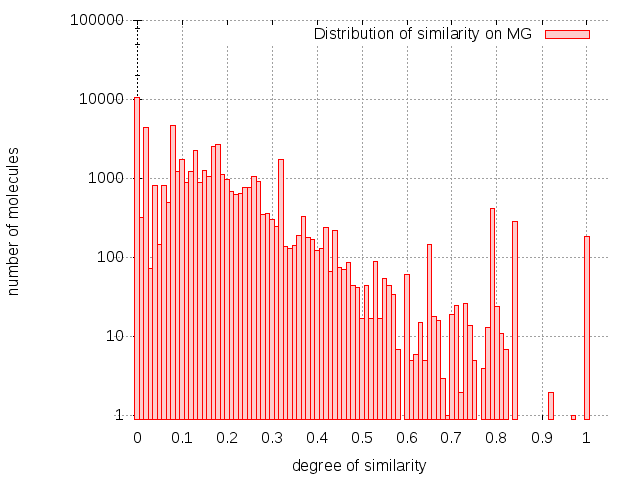}  
 \includegraphics[width=2.6in]{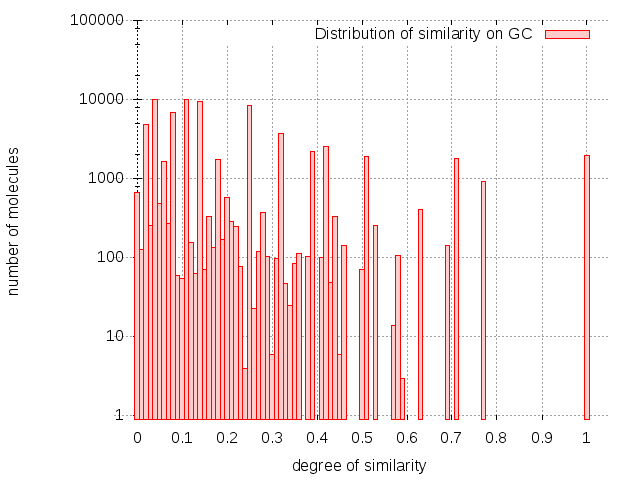} 
\end{center}
\caption{Distribution of similarity of Cholesterol on molecular graphs (MG) and on graphs of cycles (GC).}

\label{cholesterol-dis}
 \end{figure}

All the molecules in MG with a similarity degree in range $[0.8 , 1.0[ $ have a degree of $1.0$ in GC. But some molecules with a degree of $1.0$ in GC are not similar to cholesterol because we accept the fact that cycles of length $6$ are similar to those of length $5$. This parameter may be adjusted on small molecules (number of atoms). 

 \begin{figure}[H]
 \label{similarchole}
 \begin{center}
\includegraphics[width=2.5in]{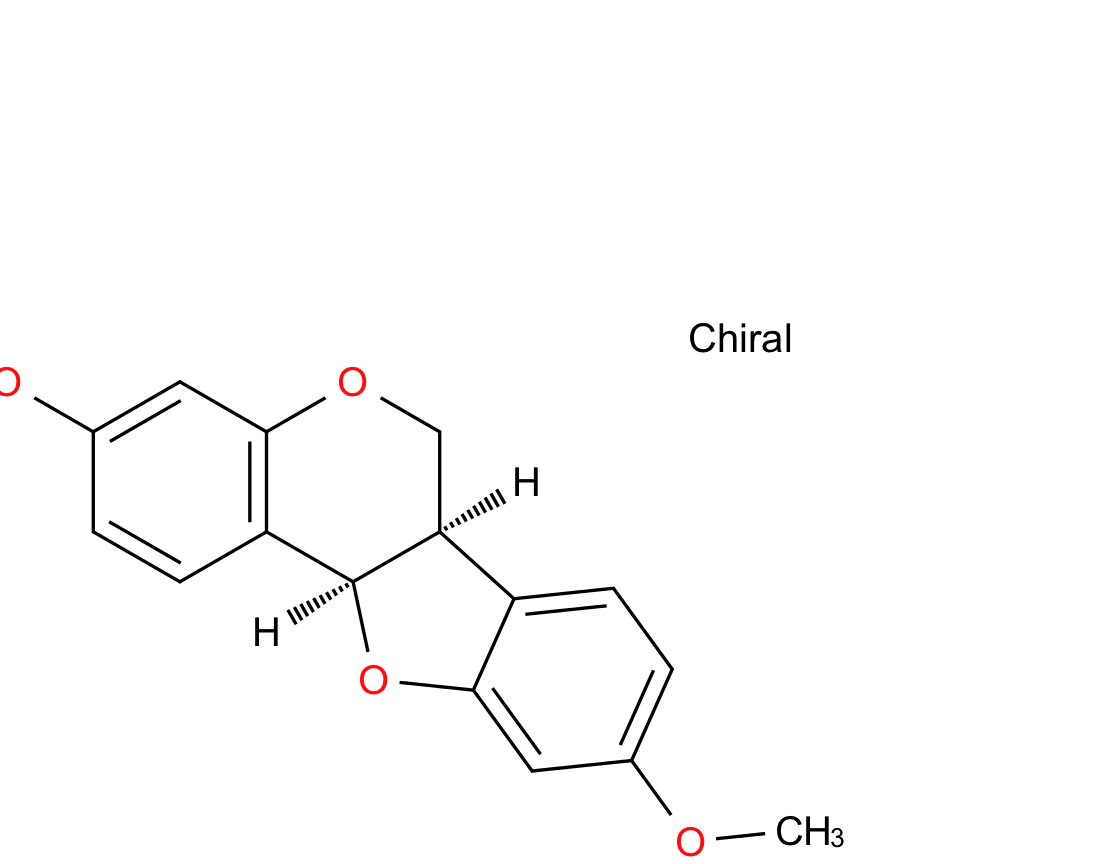}  
\includegraphics[width=2.5in]{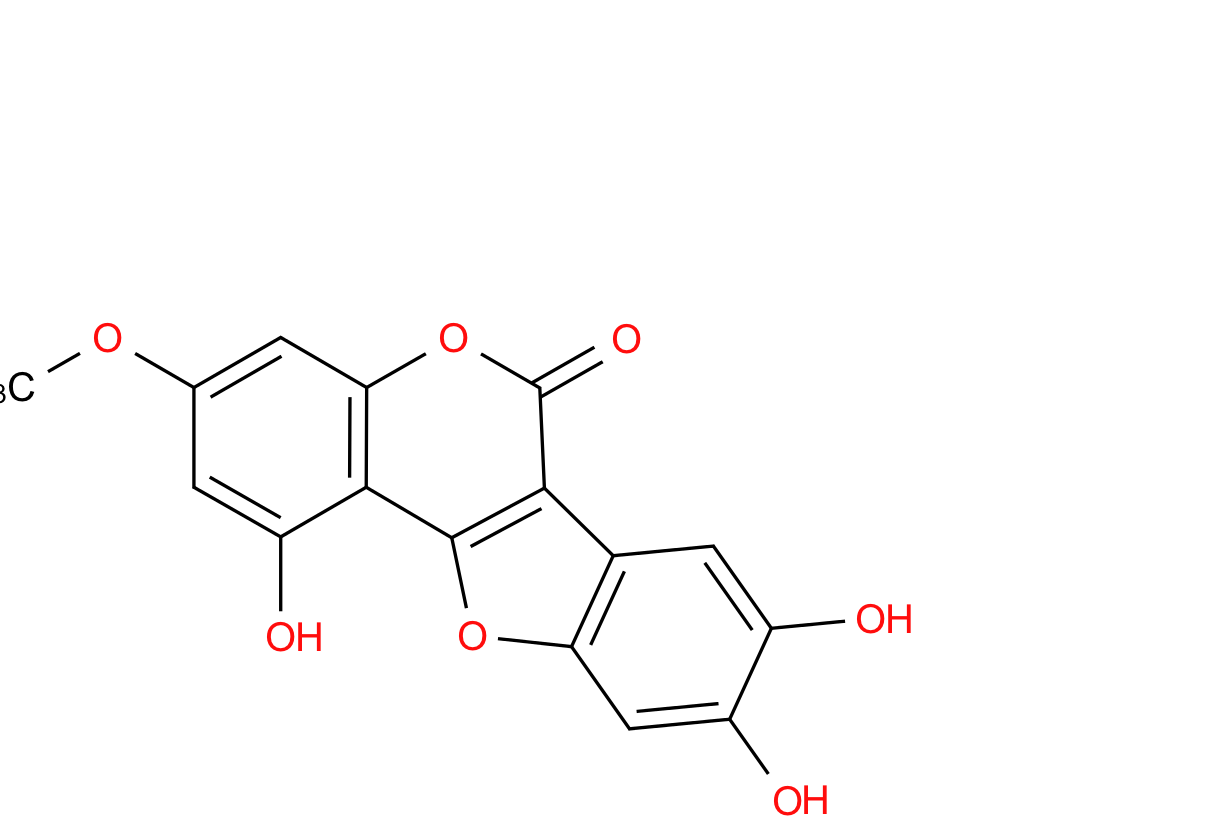} 
\end{center}
\caption{wedelolactone and $(-)-$medicarpin }
 \end{figure}
 
 In GC, wedelolactone and $(-)-$medicarpin are similar with $1.0$ with variation of cycles accepted . If we change this parameter such that cycle of different lengths cannot match, this two molecules are no more similar to cholesterol (similarity degree of $0.3$).

\subsubsection{Manzamine A}
 
 Manzamine A doesn't have many similar molecules in database CHEBI apart from molecules of the same family. This molecule has a particular structure with a cycle of length $13$ connected to small cycles.

\begin{figure}[H]
 \label{manzamine}
 \begin{center}
\includegraphics[width=5.4in]{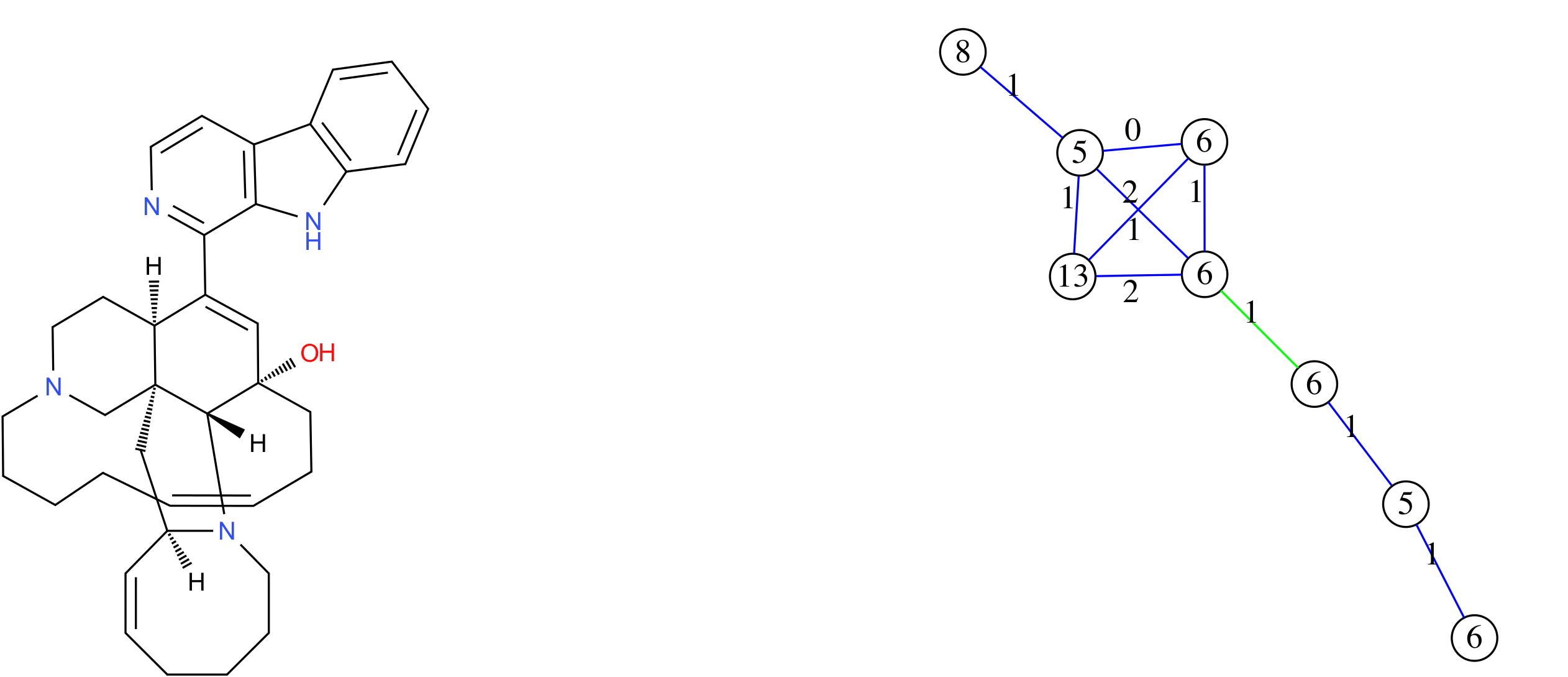}
\end{center}
\caption{Molecular graph and graph of cycles of Manzamine A. }
 \end{figure}
 
 GC similarity was able to catch the molecules of the same family with a degree of similarity of $1.0$ while MG wasn't able to find one of them. The computation time was fixed to $30$ secondes for this molecule. Over $31771$ of $90130$ molecules where not computed for MG ($35.2 \%$).
 \begin{figure}[H]
 \begin{center}
\includegraphics[width=2.6in]{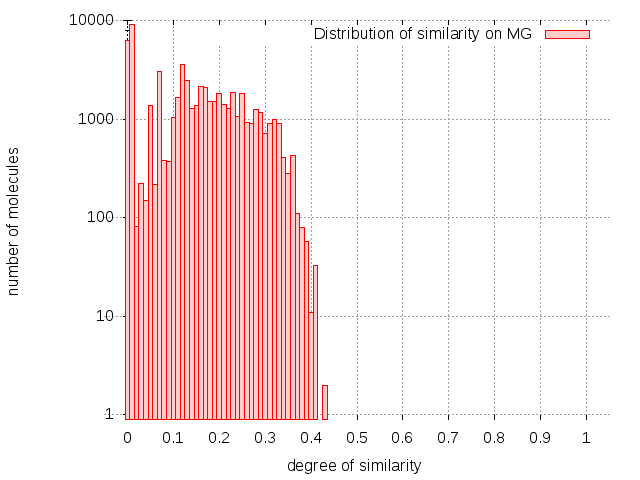}  
 \includegraphics[width=2.6in]{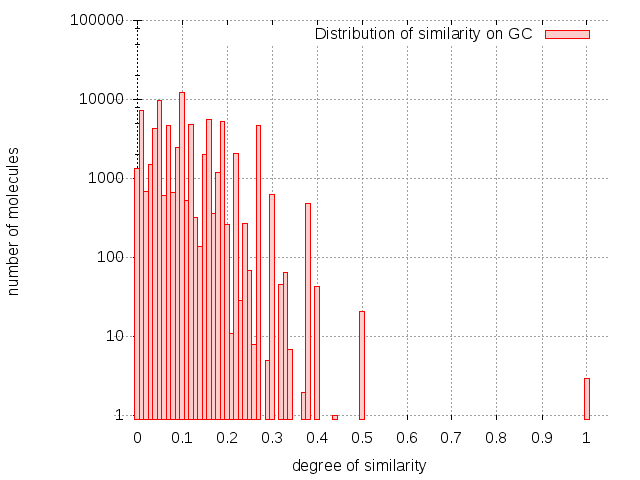} 
\end{center}
\caption{Distribution of similarity of Manzamine A on molecular graphs (MG)  and on graphs of cycles (GC).}
\label{manzamine-dis}
 \end{figure}
 
In GC, molecules similar to manzamine A with a degree of $0.5$ have a small similarity with the structural part of Manzamine A but are not really relevant. The method MG doesn't give any satistying results; the first molecules are not similar to Manzamine A.
 
 \begin{figure}[H]
 \label{similarmanza}
 \begin{center}
\includegraphics[width=2.4in]{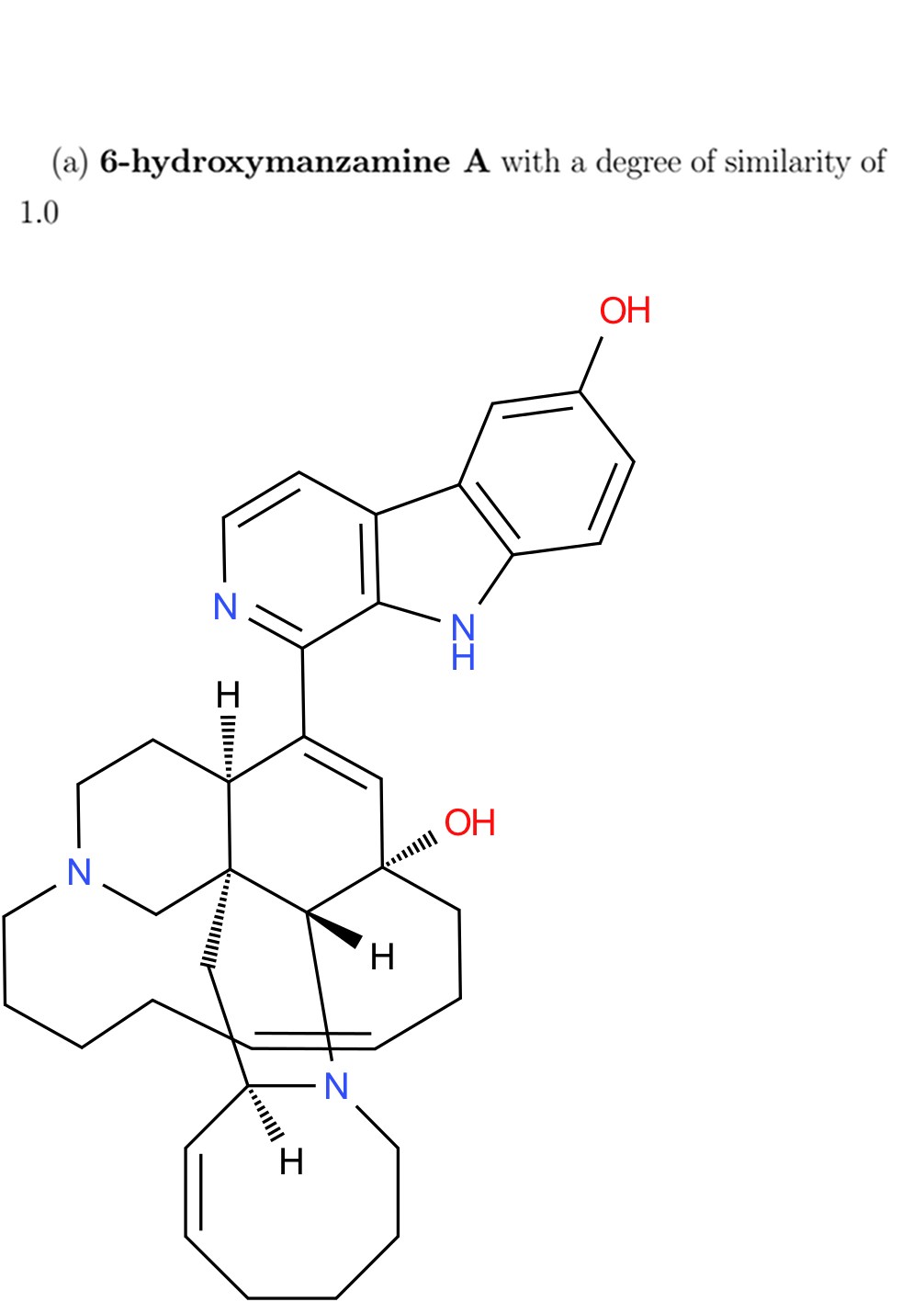}  
\includegraphics[width=2.4in]{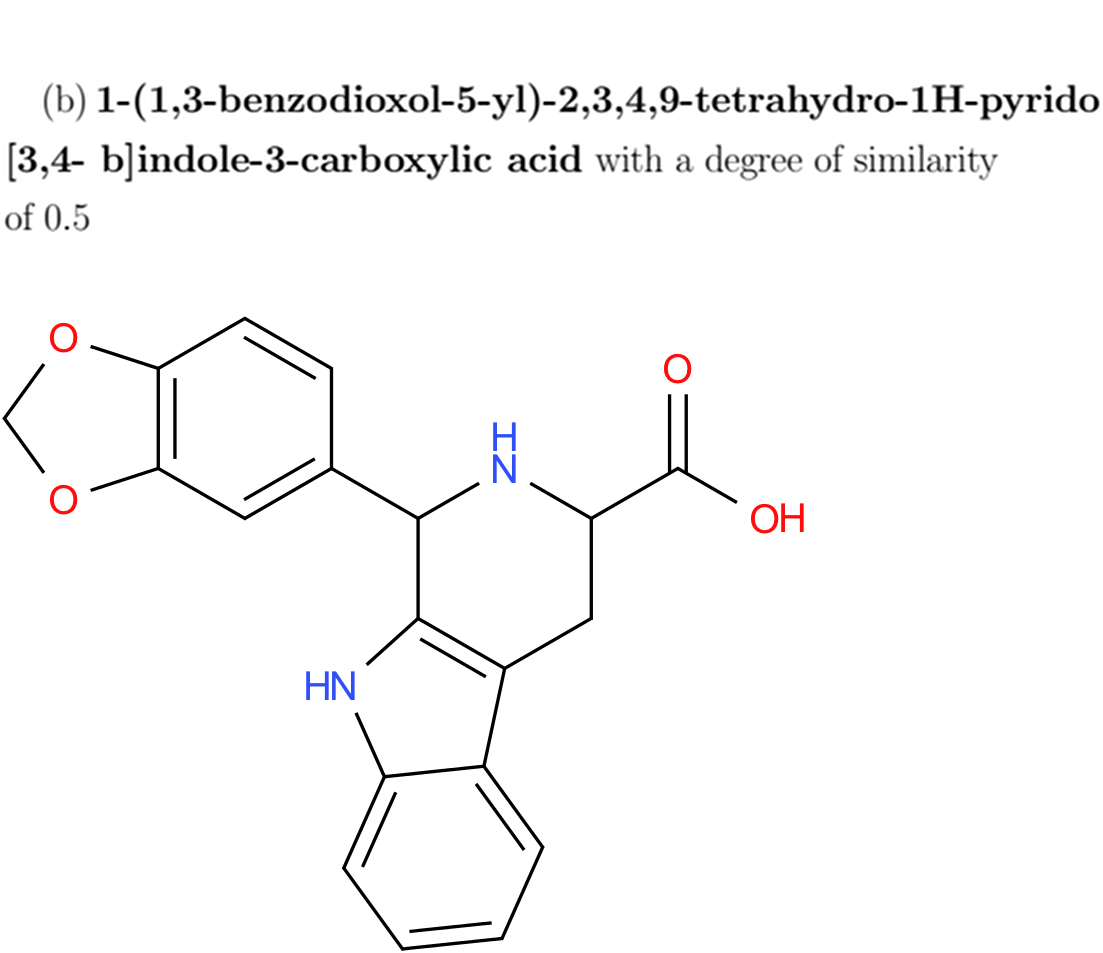} 
\end{center}
\caption{Results of similarity for Manzamine  A on GC }
 \end{figure}

\subsubsection{Brevetoxine A}
 
The structural part of Brevetoxine A is a chain a cycles. It particularity is the length of it cycles ($5$, $6$ ,$7$, $8$ and $9$) with two cycles sharing $0$ or $1$ common edge in the molecular graph.
 
\begin{figure}[H]
 \label{brevetoxine}
 \begin{center}
\includegraphics[width=5.4in]{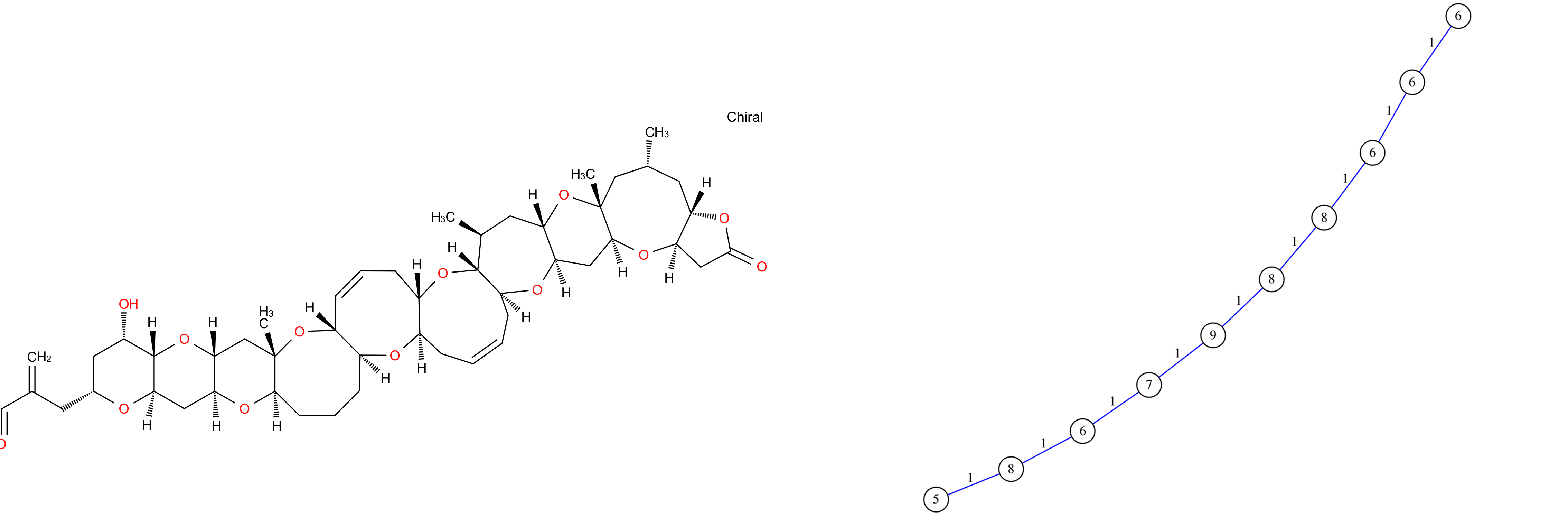}
\end{center}
\caption{Molecular graph and graph of cycles of Brevetoxine A. }
 \end{figure}
 
 In GC, we have $3$ categories :
 \begin{itemize}
 
 \item $5 $ molecules are similar with a degree upper than $0.64$. They are member of the same family with Brevetoxin A.
 \item $2$ molecules are similar with a degree equal to $0.47$ are partially similar. Their GCs are subgraphs of the GC of Brevetoxine A.
 \item The rest of molecules with a degree of similarity lower than $0.4$ are not similar to Brevetoxin A.
 \end{itemize}
 \begin{figure}[H]
 \begin{center}
\includegraphics[width=2.6in]{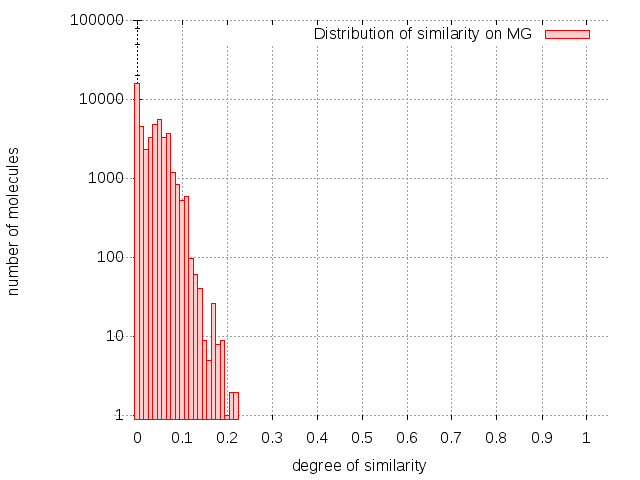}  
 \includegraphics[width=2.6in]{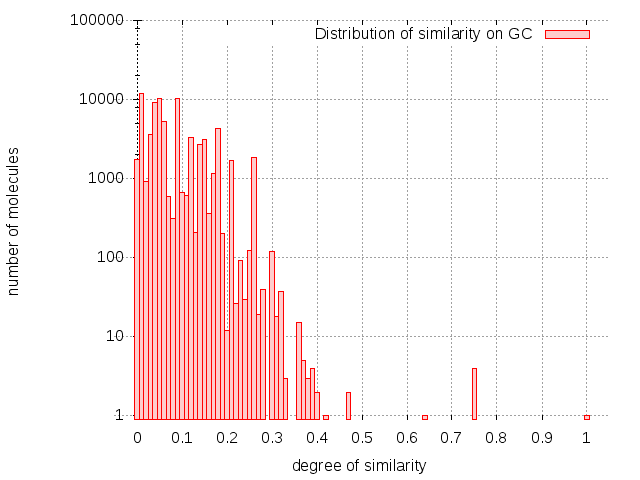} 
\end{center}
\caption{Distribution of similarity of Brevetoxine A on molecular graphs (MG)  and on graphs of cycles (GC).}
\label{brevetoxine-dis}
 \end{figure}

For MG, the parameter of time was fixed to $40$ seconds.  Over $43 237$ of $90 130$ molecules where not computed for MG ($47.9 \%$). The first molecule on top $20$ is not similar to Brevetoxine A and has a degree of similarity equals to $0.2$. 

\begin{figure}[H]
 \label{similarbreve}
 \begin{center}
\includegraphics[width=2.6in]{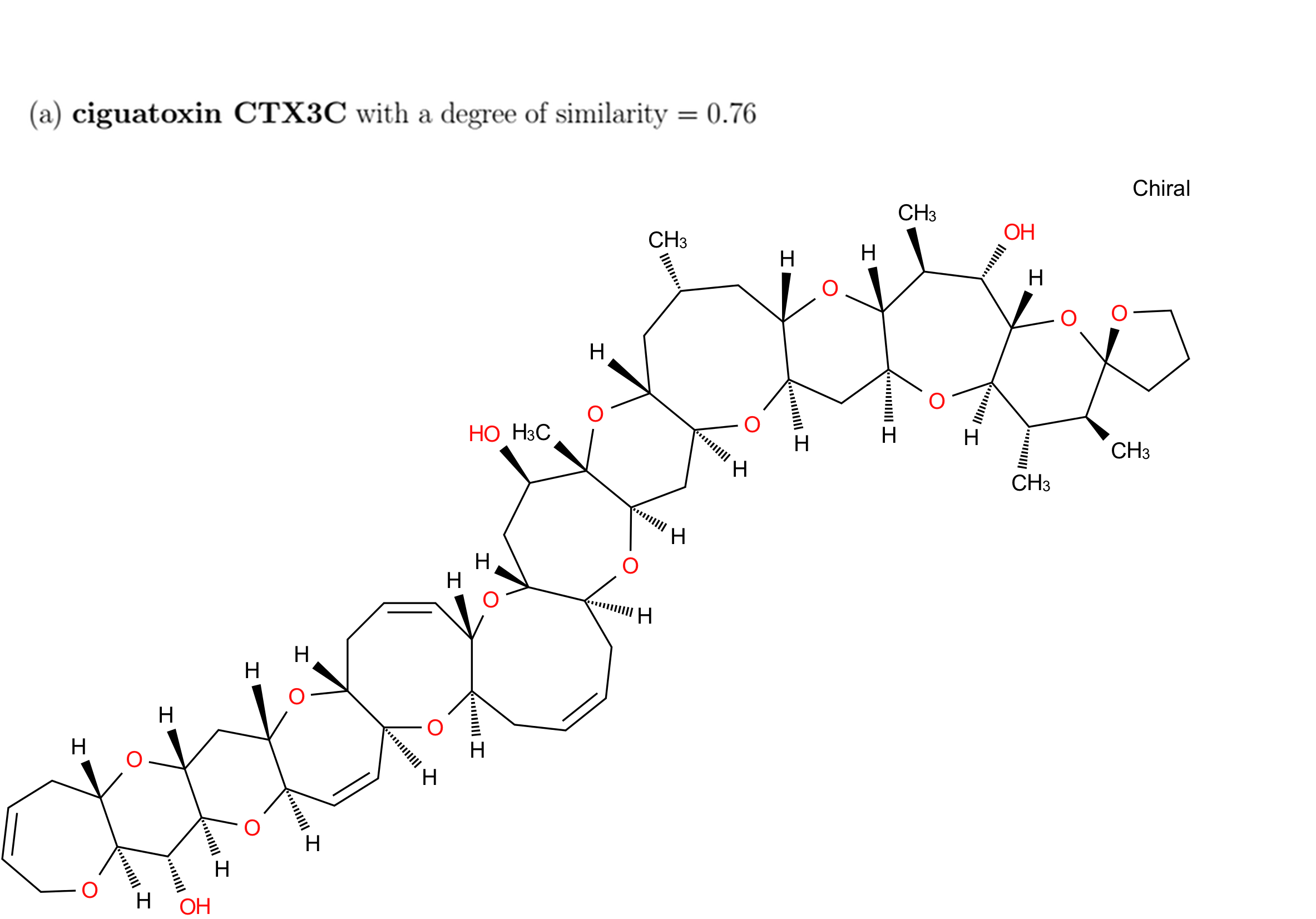}  
\includegraphics[width=2.6in]{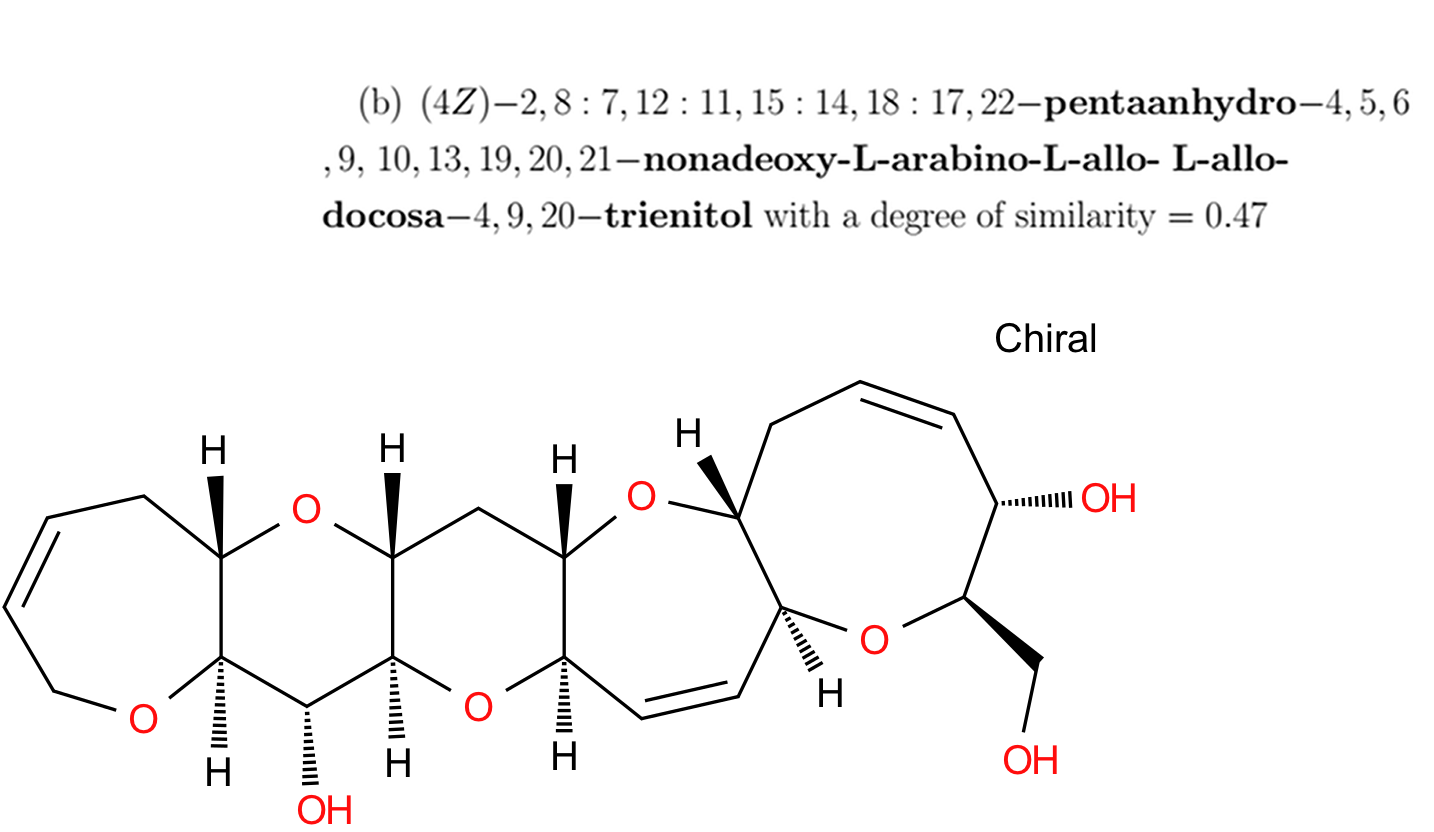} 
\end{center}
\caption{Results of similarity for brevetoxine A on GC }
 \end{figure}
 \section{Conclusion}

 The objective of this article was to propose a quickly computable measure of molecule pairwise similarity that can distinguish between pairs of molecules whose core structures (i.e., the interconnection of elementary cycles) are similar. The experiments carried out lead to several conclusions:

\begin{itemize}

\item First, the graph of cycles approach allows the similarity of all the pairs of molecules to be calculated in a reasonable time, while the molecular graph approach can require unrealistic execution times, especially for pairs of similar molecules.

\item Secondly, the similarity measures obtained by the cycle graph approach distinguish the similar pairs of molecules, and not just the isomeric molecules of the target molecules.

\item Finally, the proposed approach discriminates well the molecules very, little or not similar to a target molecule, while the approach by molecular graphs, when it can calculate similarity measures, is less discriminating.

\end{itemize}

These experiments therefore show that the proposed approach allows decision support for the determination of chemicals reactions that can be used to synthesize this target molecule from available compounds.

An extension of the proposed approach would be to be able to set the size of the cycles (parameter $j$) according to the characteristics of the molecular graph. Indeed in many cases, taking into account cycles of too large size can distort the similarity measurement because these cycles do not reflect the core structure of the molecules, while in some other cases, taking into account of such cycles is necessary to take all the core structures into account. It seems that the differentiation between these two cases depends, at least in part, on topological properties of the molecular graph, which requires further studies. Finally, the use of other metrics of similarity than the resolution of the MCES problem, for example the use of an editing distance between the cycle graphs, could also be considered.

\bibliographystyle{abbrv}
\nocite{*}
\bibliography{graph-cycles}
\end{document}